\documentclass[letterpaper,11pt]{article}

\usepackage{ucs}
\usepackage[utf8x]{inputenc}
\usepackage{graphicx}
\usepackage{amsfonts}
\usepackage{dsfont}
\usepackage{amssymb}
\usepackage{amsmath}
\usepackage{amsthm}
\usepackage{enumerate}
\usepackage{stmaryrd}
\usepackage{fullpage}
\usepackage{ifthen}
\usepackage{subfigure}
\usepackage{epic}
\usepackage{authblk}
\usepackage{textcomp}
\usepackage[small]{caption}

\usepackage[hypertexnames=false,colorlinks=true,linkcolor=blue,citecolor=blue]{hyperref}
\usepackage[numbers,comma,square,sort&compress]{natbib}
\usepackage[letterpaper,text={7in,9in},centering]{geometry}

\usepackage{color}
\usepackage{titlesec}
\graphicspath{{eps/}{pdf/}}

\newcommand{\be}{\begin{equation}}
\newcommand{\ee}{\end{equation}}
\newcommand{\bqs}{\begin{equation*}}
\newcommand{\eqs}{\end{equation*}}

\renewcommand{\O}{\mathcal{O}}

\newcommand{\e}{\epsilon}

\numberwithin{equation}{section}

  {\begin{trivlist}\item[]{\bf Hypothesis #1 }\em}{\end{trivlist}}

\theoremstyle{plain}
\newtheorem{theorem}{Theorem}[section]

\newtheorem{rmk}[theorem]{Remark}

\newcommand{\me}{\mathrm{e}}

\newcommand{\sv}{\mathbf{s}}
\newcommand{\jv}{\mathbf{j}}
\newcommand{\ev}{\mathbf{e}}
\newcommand{\vh}{\mathbf{h}}

 {\begin{trivlist}\item[]\textbf{Proof#1 }}%
 {\hspace*{\fill}$\rule{0.3\baselineskip}{0.35\baselineskip}$\end{trivlist}}

\title{Epidemic spreading on complex networks as front propagation into an unstable state  }

\author[1]{Ashley Armbruster}
\author[2,3]{Matt Holzer}
\author[4]{Noah Roselli}
\author[5]{Lena Underwood}
\affil[1]{\small Frostburg State University, Frostburg, MD, USA}
\affil[2]{\small Department of Mathematical Sciences, George Mason University, Fairfax, VA, USA }
\affil[3]{\small Center for Mathematics and Artificial Intelligence (CMAI), George Mason University, Fairfax, VA, USA}
\affil[4]{\small New Jersey Institute of Technology, Newark, NJ, USA}
\affil[5]{\small Macalester College, St. Paul, MN, USA}

\begin{document}
\maketitle

\begin{abstract} We study epidemic arrival times in meta-population disease models through the lens of front propagation into unstable states.  We demonstrate that several features of invasion fronts in the PDE context are also relevant to the network case.  We show that the susceptible-infected-recovered model on a network is linearly determined in the sense that the arrival times in the nonlinear system are approximated by the arrival times of the instability in the system linearized near the disease free state.  Arrival time predictions are extended to general compartmental models with a susceptible-exposed-infected-recovered model as the primary example.   We then study a recent model of social epidemics where higher order interactions lead to faster invasion speeds.  For these pushed fronts we compute corrections to the estimated arrival time in this case.  Finally, we show how inhomogeneities in local infection rates lead to faster average arrival times.

\end{abstract}



\section{Introduction }

The study of global disease spread across complex networks has been the focus of a great deal of research over the past several decades; see \cite{barrat08,brockmann13,colizza06,kiss17,pastor15,taylor15} for a survey of many of the models and methods employed.  Meta-population models comprise one sub-class of models
 where the disease dynamics at each locality are assumed to obey some compartmental model (SIR for example) and movement of individuals between localities is modeled by diffusion on a complex network; see for example  \cite{brockmann13,rvachev}.  This leads to a high dimensional system of ODEs of reaction-diffusion type.  Among the questions that one is interested in are arrival times: given that disease originates in one city how long does it take to appear in some other city?  For reaction-diffusion PDEs,  instability spreading involves the formation of traveling fronts and arrival times are inversely proportional to the speed selected by these fronts; see for example \cite{bramson83,ebert00}.  It is a powerful, albeit perhaps peculiar, fact that often times the speed of the front in the nonlinear PDE is the same as the spreading speed of localized initial conditions in the PDE linearized about the unstable state; see \cite{aronson78,vansaarloos03}.  This fact was exploited in \cite{chen18} to derive arrival times estimates based upon linearization near the unstable, disease free state.  The purpose of the current study is to exploit this analogy further and demonstrate that several features of front propagation into unstable states for PDEs are also relevant to meta-population epidemics spreading on complex networks. 
 
The meta-population model that we first consider is the following one;  see \cite{brockmann13,rvachev},
\begin{eqnarray}
\partial_{t}s_{n}&=&-\alpha s_{n}j_{n}+ \gamma\sum_{m\neq n}P_{nm}(s_{m}-s_{n}) \nonumber \\
\partial_{t}j_{n}&=&\alpha s_{n}j_{n}-\beta j_{n}+ \gamma\sum_{m\neq n}P_{nm}(j_{m}-j_{n}) \nonumber \\
\partial_{t}r_{n}&=&\beta j_{n}+ \gamma\sum_{m\neq n}P_{nm}(r_{m}-r_{n}) 
.\label{eq:main} \end{eqnarray}

Here $s_n$, $j_n$ and $r_n$ denote the susceptible, infected and recovered proportion of the population residing at node (city) $n$.  The dynamics of these variables is assumed, for the moment,  to obey a standard SIR model at each node with infection rate $\alpha$ and recovery rate $\beta$.  The nodes are connected by edges described by the row stochastic adjacency matrix $\mathrm{P}$.  Following \cite{brockmann13} we think of these edges are describing airline transportation routes connecting cities with the values in the matrix representing a normalized magnitude of passenger transport along each edge.  {\color{black} The matrix $\mathrm{P}$ is assumed to be row stochastic so that the total population at each node is constant in time and it is only the proportion of the susceptible, infected and recovered population that varies.}  The parameter $\gamma$ is the diffusion constant and we, crucially, will assume that it is small (see again \cite{brockmann13} for estimates of $\gamma$ for the global airline network).  We note that the number of nodes in typical realizations of (\ref{eq:main}) is large (on the order of thousands for the airline transportation network) and the corresponding parameter space is also high dimensional due to the large number of non-zero entries in $\mathrm{P}$.

Brockmann and Helbing \cite{brockmann13} studied (\ref{eq:main}) with the goal of identifying the role of geographically non-local transportation routes in the global spread of epidemics.  Focusing on airline transportation networks, their influential idea was to consider the epidemic process as a front propagation with respect to some {\em effective distance}, $D_{\mathrm{eff}}(\mathrm{P})$,  that can be derived from the connectivity matrix $\mathrm{P}$.  They then predict the arrival time of the disease at a city as the ratio $T_a=\frac{D_{\mathrm{eff}}(\mathrm{P})}{v_{\mathrm{eff}(\alpha,\beta,\kappa)}}$\footnote{$D_{\mathrm{eff}}(\mathrm{P})$ as defined in \cite{brockmann13} is defined by first computing an effective distances between connected nodes defined as $1-\log(\mathrm{P}_{mn})$.  Then, for any two notes that are not connecting the effective distance is defined as the minimal sum of effective distances along all paths connecting the two nodes.} so that arrival times are linearly related to the effective distance. Here $v_{\mathrm{eff}}(\alpha,\beta,\gamma,\kappa)$ is the effective velocity which is assumed to be a function of the dynamical parameters in the model and $\kappa$, an invasion threshold.    Key to this idea is the fact that the effective distance depends only on the structure of $\mathrm{P}$.  As such, the distance prediction is agnostic in regards to the particular disease model considered.  Estimates for the coefficients in the real  world $\mathrm{P}$ are obtained in \cite{brockmann13} and comparisons with data of observed arrival times in historical epidemics are considered which reveal a general linear trend between arrival times and these effective distances.  One drawback of the effective distance computed in \cite{brockmann13} is that it assumes there is a single dominant pathway of infection between the origin city and any other city in the graph.  Modifications of this effective distance to account for multiple pathways of infection are presented in \cite{iannelli17}.  In addition to \cite{brockmann13}, a number of other authors have considered the dynamics of global disease spread through the lens of front propagation; see for example \cite{belik11,besse21,chen18,gautreau07,gautreau08,hindes13,hoffman19,hufnagel04}.

A remarkable feature of fronts propagating into unstable states in the PDE context\footnote{In fact, this phenomena occurs more generally for spatially extended systems such as lattice dynamical systems or systems with non-local diffusion in both discrete and continuous time; see \cite{hindes13,weinberger02} among others for examples.} 
is that their speed often equals the spreading speed of localized disturbances in the system linearized about the unstable state.   {\color{black} This phenomena often occurs in systems where the nonlinearity suppresses growth, as is the case in (\ref{eq:main}).}  Such fronts are referred to as {\em pulled}; see \cite{vansaarloos03} or linearly determinate; see \cite{weinberger02} as they are driven by the instability ahead of the front interface and their speed is determined from the linearization near the unstable state.   This is a powerful tool as it allows for the computation of a quantity of interest in a high (or infinite) dimensional nonlinear system via a linear equation.  This forms the basis of the approach in \cite{chen18} where arrival time estimates are derived for (\ref{eq:main}) by computing arrival times in the system linearized near the disease free state.  The goal of the present study is to exploit this analogy between the dynamics of  reaction-diffusion equations like (\ref{eq:main}) and their PDE counterparts to make qualitative predictions regarding the effects of arrival times where various modifications of (\ref{eq:main}) are made.  {\color{black} Our main results are qualitative in nature and can be summarized as follows:
\begin{itemize}
\item For systems with local dynamics described by SIR or SEIR we derive explicit arrival times estimates based upon linearization near the unstable state that reveal, in the limit $\gamma\to 0$, how arrival times depend on model parameters such as local infection rates, local recovery rates and mobility network weights.  At leading order in $\gamma$, the effective distance between nodes is shown to be the graph distance $d$ while the effective velocity is proportional to $-\frac{1}{\log(\gamma)}$.   Network properties influence arrival times at $\O(1)$ in $\gamma$ where the key quantity is the random walk probability of traversing between the two cities in the minimal number of steps.  
\item We show, by way of an example, that linear arrival times are not good estimates for all systems.  This example occurs for a model for which the nonlinearity enhances growth of the local infection and we explain the mechanism by which this leads to faster arrival times drawing an analogy to pushed fronts in spatially extended systems.  Based upon an analysis of the local dynamics, we derive arrival time estimates and compare them with numerical simulations. 
\item We show that inhomogeneities in local reaction rates lead to faster arrival times on average.   We attribute this to the following mechanism based upon the linear arrival times estimates for the homogeneous SIR model: increasing infection rates leads to a decrease in arrival times at  $\O\left(-\log(\gamma)\right)$ while decreasing the random walk probability between two nodes decreases arrival times at $\O(1)$.  Thus, if two cities are connected by at least one shortest path consisting of cities with higher than average infection rates we expect an overall decrease in the arrival time.  
\end{itemize} 

}


It bears mentioning that if one had reliable estimates for the parameters in (\ref{eq:main}) -- the infection rate $\alpha$, the recovery rate $\beta$ and the coefficients of the mobility matrix $\mathrm{P}$ then to estimate arrival times one could simply numerically solve the system of ODEs in (\ref{eq:main}).  In fact, this would serve as a forecast for the entire course of the epidemic.  More broadly, there are a number of sophisticated tools for the forecasting of epidemics; see for example the GLEAM simulator \cite{balcan10,vandenbroeck11}.  In this light, our goal in this work is not epidemic forecasting but instead is to present qualitative predictions for how arrival times depend on system features and to strengthen the relationship between the dynamics of (\ref{eq:main}) and the theory of invasion fronts in PDEs or other spatially extended systems which will, in turn, help inform researchers making epidemic forecasting.  {\color{black} Qualitative statements are particularly useful for systems with a high dimensional parameter spaces, as is the case with (\ref{eq:main}). }

We discuss some limitations of the present study.  Most arrival time estimates that we provide are obtained in the limit of small $\gamma$.  In particular, our explicit arrival time estimates will require $\gamma$ to be asymptotically smaller than various quantities including the instability parameter $\alpha-\beta$ and the coefficients of the mobility matrix $\mathrm{P}$.    While $\gamma$ is naturally expected to be small, once again see \cite{brockmann13}, it is not expected that these conditions will hold generally for real world transportation networks.  While some of these deficiencies could be likely remedied by a more detailed analysis we do not pursue such estimates here. Another interesting avenue for research is to study how well the arrival times estimates for the deterministic model (\ref{eq:main}) reflect those in stochastic versions of epidemic spread; we point the reader to \cite{jamieson20,jamieson21} for recent work in this direction.  

For the purposes of illustrating our main results we will perform numerical simulations of (\ref{eq:main}) on a version of the world wide airline transportation network obtained from \cite{openflights}.  This is a historical snapshot from June 2014.  There are $N=3304$ airports and the network has $19,082$ edges representing one or more flights connecting two cities.  The mean degree is $11.53$.  We will use this network to illustrate some of our results and arrival time estimates, but we do not pursue a full numerical investigation.  {\color{black} For the purposes of numerical simulations, we do not attempt to construct accurate approximations for the flux matrix $\mathrm{P}$ as was done in \cite{brockmann13}.   Let $\mathrm{A}$ be the symmetric adjacency matrix for the airline transportation network from \cite{openflights} where the entry $\mathrm{A}_{nm}$ equals  $1$ if there exists a flight connecting cities $n$ and $m$ and $0$ otherwise.  Let $\mathrm{D}$ be the diagonal degree matrix.  Then we will take $\mathrm{P}=\mathrm{D}^{-1}\mathrm{A}$ for simplicity.   For future reference, we define the graph distance $d_{mn}$ as the minimum length path between the node $n$ and $m$.  When the origin node $n$ is fixed we will shorten this to $d_m$.  }

The rest of the paper is organized as follows.  In Section~\ref{sec:reviewAT}, we review and motivate the arrival time estimate of \cite{chen18}.  In Section~\ref{sec:SEIR}, we extend this arrival times estimate to a susceptible-exposed-infected-recovered (SEIR) model.  In Section~\ref{sec:pushed}, we show that the linear arrival time estimate is no longer valid in a model of social epidemics that incorporates higher order interactions between individuals but are able to make corrections to the arrival time estimate to yield approximations.  In Section~\ref{sec:inhomo}, we study the effect of inhomogeneous infection rates on arrival times and argue that this will decrease arrival times on average.  

\section{Arrival time estimates via linearization near the disease free state} \label{sec:reviewAT}
In this section, we review the arrival time estimate presented in \cite{chen18}.  {\color{black} We assume that the disease originates in city $n$ with the initial infected proportion $j_n(0)=\chi_0$ so that $s_n(0)=1-\chi_0$.  We are interested in nonlinear arrival times $t_{mn}$ defined as the  minimal time at which $j_m(t)$ exceeds some threshold $\kappa$.  The primary purpose of this section is to review how estimates for $t_{mn}$ can be obtained by linearizing near the unstable, disease-free state. } 

The arrival time estimate in \cite{chen18} is predicated on the fact that (\ref{eq:main}) is linearly determined; see \cite{weinberger02}, which informally means that the linear arrival times will be a good prediction for the nonlinear arrival times.  Thus, the first step is to linearize (\ref{eq:main}) near the unstable state (we will neglect the recovered population from here forward) to obtain the following system of linear equations expressed in vector form,
\begin{eqnarray*} \sv_t &=& -\alpha \jv +\gamma \left( \mathrm{P}-\mathrm{I}\right) \sv  \\
\jv_t &=&  (\alpha-\beta) \jv +\gamma \left( \mathrm{P}-\mathrm{I}\right) \jv .
\end{eqnarray*}

The $\jv$ component decouples and can be solved using the matrix exponential,
\be \jv(t)=\chi_0 \me^{(\alpha-\beta-\gamma)t}\me^{\gamma \mathrm{P} t} \boldsymbol{\delta}_n, \label{eq:linsol} \ee
{\color{black} where $\boldsymbol{\delta}_n$ is the standard Euclidean basis vector and $\chi_0$ is the initial infected proportion residing in city $n$.  The arrival time in city $m$ is defined as the first time where the infected proportion of the population exceeds a threshold $\kappa$ and is therefore the smallest positive solution of 
\be j_m(t_{mn})=\kappa. \label{eq:ATdefn} \ee
Let $\tau_{mn}$ be an estimate for $t_{mn}$ obtained by setting the $m$-th component of (\ref{eq:linsol}) equal to $\kappa$.  To obtain this estimate,  project the solution in (\ref{eq:linsol}) onto $\boldsymbol{\delta}_m$ to extract the infected proportion at the $m$-th node.  Then we wish to solve 
\be \kappa=\boldsymbol{\delta}_m^T \chi_0 \me^{(\alpha-\beta-\gamma)\tau_{mn}}\me^{\gamma \mathrm{P} \tau_{mn}} \boldsymbol{\delta}_n. \label{eq:linAT} \ee}
To exploit the smallness of the parameter $\gamma$, the matrix exponential is expanded as a series,
\[ \boldsymbol{\delta}_m^{T}\me^{\gamma \mathrm{P} \tau_{mn}} \boldsymbol{\delta}_n=\sum_{k=0}^\infty \gamma^k\frac{\tau_{mn}^k}{k!} \boldsymbol{\delta}_m^T \mathrm{P}^k  \boldsymbol{\delta}_n. \]
The coefficients $\boldsymbol{\delta}_m^T \mathrm{P}^k \boldsymbol{\delta}_n$ are random walk probabilities for a walker traveling from city $m$ to city $n$ in $k$ steps.  As such, all these terms are zero up to $k=d_m$ where we recall that $d_m$ is the graph distance between the origin city $n$ and the arrival city $m$.  Let $\rho_m=\boldsymbol{\delta}_m^T P^{d_m} \boldsymbol{\delta}_n$.  Now, for $\gamma $ sufficiently small we assume that the leading order term in the sum dominates and we obtain a leading order expression for $\tau_{mn}$ by solving 
\be \kappa=\frac{\chi_0\rho_m}{d_m!} \gamma^{d_m} \tau_{mn}^{d_m} \me^{(\alpha-\beta-\gamma)\tau_{mn}}.  \label{eq:themainlin} \ee
The solution of this equation can be expressed in terms of the Lambert-W function, and we obtain the arrival time estimate
\be \tau_{mn}= \frac{d_m}{\alpha-\beta}W\left( \frac{1}{\gamma}\frac{\alpha-\beta}{d_m} \left( \frac{d_m! \kappa}{\rho_m \chi_0} \right)^{1/d_m}\right). \label{eq:linATW} \ee
Expanding the Lambert-W function we obtain  
\be \tau_{mn} = -\frac{d_m}{\alpha-\beta}\log(\gamma)-\frac{d_m}{\alpha-\beta}\log (-\log(\gamma))-\frac{d_m}{\alpha-\beta}\log\left(\frac{d_m}{\alpha-\beta}\left(\frac{\rho_m \chi_0}{d_m! \kappa}\right)^{1/d_m}\right) +o(1), \label{eq:linATex} \ee
{\color{black} where $o(1)$ represent terms that go to zero as $\gamma\to 0$; see again \cite{chen18}.  Note that in the $\O(1)$ terms we write the argument of the logarithm so that it is clear that larger values of the random walk probability $\rho_m$ lead to faster arrival times.  }

{\color{black}  The primary take-away from (\ref{eq:linATex}) is that only two network features are relevant for the determination of arrival times (in the limit as $\gamma\to 0$) and are i) the graph distance between the origin and arrival cities and ii) the random walk probability of traversing between these two cities in the minimal number of steps.  We remark further that, to leading order, the effective distance is simply the graph distance between the nodes $n$ and $m$ while the effective velocity is $\frac{\alpha-\beta}{-\log(\gamma) }$.  This is consistent with the spreading speed of instabilities along one dimensional lattices; see for example \cite{hoffman19}. We also note that if $\O(1)$ terms are involved then it is no longer possible to separate the arrival times into a ratio of a network-dependent effective distance and a dynamics-dependent effective velocity.  Once again, we emphasize that these are asymptotic estimates and should be expected to hold in limit as $\gamma\to 0$.  For larger values of $\gamma$, we are not able to explicitly connect network properties to arrival times, although  we emphasize that numerical results suggest that linear arrival times remain good estimates for nonlinear arrival times in this case; see Section~\ref{sec:discussion}.}

As we have stressed above, the fact that these arrival time estimates are accurate stems from the fact that (\ref{eq:main}) is linearly determined.  In fact, we have the following result which proves that the linear arrival times are always a lower bound for the nonlinear arrival times.
\begin{theorem}\label{thm:super} Consider (\ref{eq:main}) with the initial conditions $s_l(0)=1$, $j_l(0)=r_l(0)=0$ for all $l\neq n$ and $s_n(0)=1-\chi_0$, $j_n(0)=\chi_0$ and $r_n(0)=0$ for some $0<\chi_0<1$.   Let $\tau_{mn}(\alpha,\beta,\gamma,\kappa,\chi_0)$ be the linearized arrival time estimate in city $m$ defined as the solution of (\ref{eq:linAT}).  Let $t_{mn}(\alpha,\beta,\gamma,\kappa,\chi_0)$ be the nonlinear arrival time of the disease at node $m$, defined by the minimum time at which 
\[ j_m(t_{mn})=\kappa. \]
Then 
\[ \tau_{mn}(\alpha,\beta,\gamma,\kappa,\chi_0)< t_{mn}(\alpha,\beta,\gamma,\kappa,\chi_0). \]
\end{theorem}

\begin{proof} The proof is a standard application of the comparison principle and was sketched in \cite{chen18}.  Let
\begin{eqnarray*}
N_S(\sv,\jv) &=& \sv_t+\alpha \sv \circ \jv -\gamma \left(\mathrm{P}-\mathrm{I}\right) \sv \\
N_J(\sv,\jv) &=& \jv_t-\alpha \sv \circ \jv+\beta \jv-\gamma \left(\mathrm{P}-\mathrm{I}\right)\jv.
\end{eqnarray*}
Here $\sv\circ\jv$ is the Hadamard, or component-wise multiplication of the vectors.  
The idea is to find functions $\bar{\sv}(t)$ and $\bar{\jv}(t)$ such that  both $N_S(\bar{\sv}(t),\bar{\jv}(t))$ and $N_J(\bar{\sv}(t),\bar{\jv}(t))$ are non-negative indicating that the temporal growth rate of the selected functions exceeds that of the true solution and therefore any initial condition for which $\sv(0)\leq \bar{\sv}(0)$ and $\jv(0)\leq \bar{\jv}(0)$  will satisfy  $\sv(t)\leq \bar{\sv}(t)$ and $\jv(t)\leq \bar{\jv}(t)$ for all $t>0$.   To begin, it is easy to see that if $\bar{\sv}(t)=1$ then $N_S(1,\jv(t))>0$.  Then we observe
\[ N_J(1,\bar{\jv}(t)) = \bar{\jv}_t-(\alpha-\beta) \bar{\jv} -\gamma \left(\mathrm{P}-\mathrm{I}\right)\bar{\jv}    . \]
Thus, if  $\bar{\jv}(t)$ is the solution of the linear equation (\ref{eq:linsol}) we have obtained a  super-solution.  The result then follows.   
\end{proof}

To fully validate that the arrival times are linearly determined would require the establishment of sufficiently sharp sub-solutions. We do not pursue this avenue of research here; although we do point to \cite{fu16,wu17} for constructions in the case of (\ref{eq:main}) posed on an infinite lattice.

{\color{black} The linearly determined arrival time estimates are compared to arrival times observed in numerical simulations in Figure~\ref{fig:comp}.  We also point out that Theorem~\ref{thm:super} does not depend on $\gamma$ being small, see Figure~\ref{fig:biggamma} }

\begin{figure}
    \centering
     \subfigure{\includegraphics[width=0.41\textwidth]{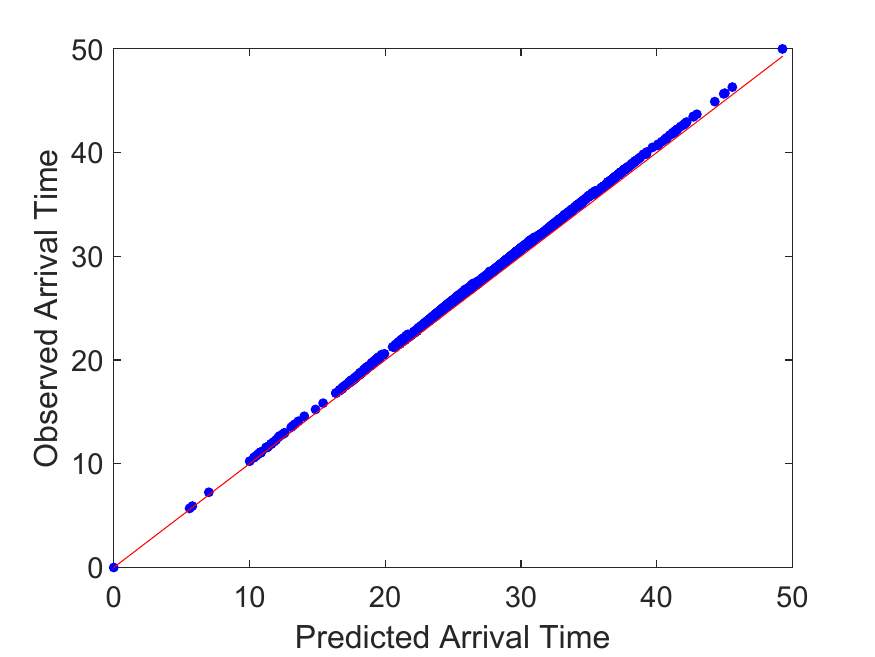}}
 \subfigure{\includegraphics[width=0.41\textwidth]{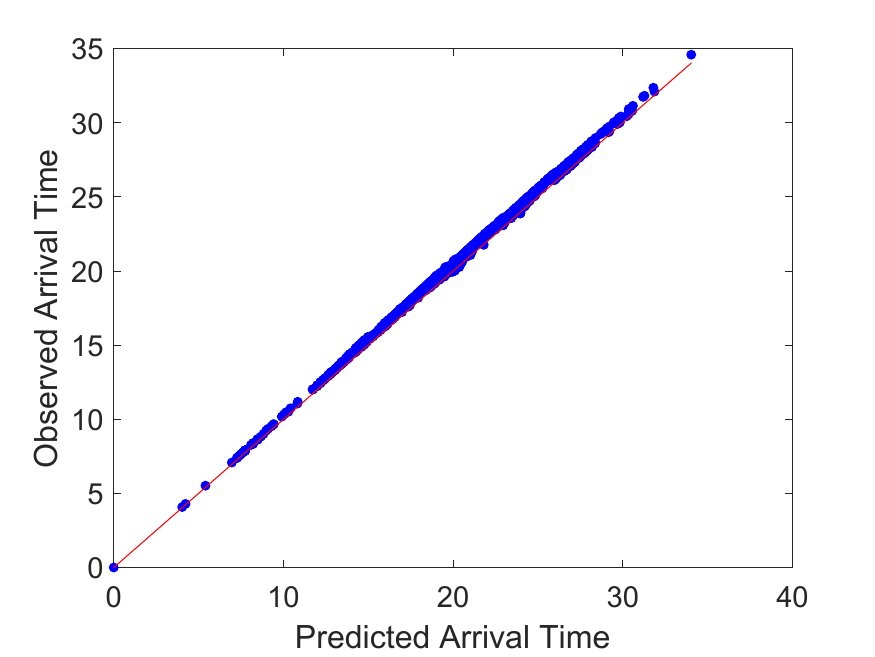}}
   \caption{Arrival time estimates given by (\ref{eq:linATW}) plotted against arrival times observed in numerical simulations of (\ref{eq:main}).  On the left, the mobility parameter $\gamma=0.001$ while on the right the mobility parameter is $\gamma=0.01$.  In both cases, the infection rate is $\alpha=1.50$ and the recovery parameter is $\beta=0.25$.   In both cases, the linear estimate is a good approximation of the nonlinear arrival times (for $\gamma=0.01$ absolute error less than $0.75$ days and relative error less than $0.04$).  We note that when $\gamma=0.001$ (left) then the observed arrival time is always greater than the predicted arrival time as is expected.  However, when $\gamma=0.01$ (right) a few cities have observed arrival times that are faster than the linear arrival time.  This does not contradict Theorem~\ref{thm:super} since only one term in the matrix exponential is used to create the linear arrival time estimate in (\ref{eq:linATW}).  If one were to include more terms in the sum, then the linear estimate would again be less than the observed nonlinear arrival time; see \cite{chen18} and Figure~\ref{fig:biggamma}.    }
    \label{fig:comp}
\end{figure}

\subsection{Alternate derivation of linear arrival time estimate} \label{sec:cascade} 
The analysis above suggests that, in the small diffusion limit, epidemic spreading in complex networks can be thought of as a cascading behavior where the epidemic spreads from the node of origination out through the network with all nodes of fixed graph distance from the origin node becoming infected at approximately the same time.    In this section, we explain how this point of view can be used to obtain  analogous arrival time estimates as in (\ref{eq:linAT}).  

{\color{black} The purpose of this section is twofold.  First and most importantly, this approach and the ideas presented here will be employed later for cases where linear determinacy fails (namely Section~\ref{sec:pushed} where faster than linear invasion speeds are observed and then in Section~\ref{sec:inhomo} where inhomogeneous infection rates lead to faster than average arrival times).  As a secondary goal, this section presents an alternate way to understand why the linear arrival time estimates derived previously are good estimates of the nonlinear arrival times.  This alternate method is more combersome than expanding the matrix exponential and relies on a number of formal calculations and therefore we do not suggest that this derivation should supplant the estimate derived by solving (\ref{eq:themainlin}).  }

To begin, without loss of generality we may assume that node $n=1$ is the node (city) at which the disease emerges.  Let $n=2$ correspond a city that is connected to the first node.  The equation for the infected population at this node is then
\[ \frac{dj_2}{dt}= \alpha s_2 j_2 -\beta j_2 +\gamma\sum_{k=1}^N \mathrm{P}_{2k} (j_k-j_2). \]
Assume that $j_2\ll 1 $, $s_2\approx 1$ and  that $j_k\ll  1$ for all $k\geq 3$.  From this it is reasonable to approximate $j_2(t)$ in the short to intermediate time by the linear equation
\be \frac{dj_2}{dt}\approx (\alpha-\beta) j_2 +\gamma \mathrm{P}_{21} j_1(t). \label{eq:j2approx} \ee 
This equation has an explicit solution 
\be j_2(t)\approx \gamma \mathrm{P}_{21} \me^{(\alpha-\beta) t}\int_0^t \me^{-(\alpha-\beta)\tau }j_1(\tau) d\tau \label{eq:j2sol} \ee 
Assuming further that $j_1(t)\approx \chi_0 \me^{(\alpha-\beta)t}$ then (\ref{eq:j2sol}) reduces to 
\[ j_2(t)\approx \gamma \mathrm{P}_{21}\chi_0 t \me^{(\alpha-\beta) t}, \]
from which we see that imposing $j_2(t_2)=\kappa$ and solving for the arrival time $t_2$ produces the same formula as in (\ref{eq:linAT}) {\color{black} and identical arrival time estimate $t_2\approx\frac{1}{\alpha-\beta} W\left(\frac{\kappa (\alpha-\beta)}{\chi_0 \gamma \mathrm{P}_{21}}\right) $.  }

Now suppose that $n=3$ is connected to $n=2$ but not connected to node $n=1$ nor any of its children (aside from node $2$).  Repeating the analysis above we can find 
\be \frac{dj_3}{dt}\approx (\alpha-\beta) j_3 +\gamma \mathrm{P}_{32} j_2(t). \label{eq:j3approx} \ee
{\color{black} We now plug in the approximation $j_2(t) \approx \gamma \mathrm{P}_{21} \chi_0 t \me^{(\alpha-\beta)t}$.  Then we obtain an approximate solution formula
\[ j_3(t)\approx \chi_0\gamma^2 \mathrm{P}_{32}\mathrm{P}_{21} \frac{t^2}{2} \me^{(\alpha-\beta)t}, \]
so that the arrival time  $t_3$, determined from setting $j_3(t_3)=\kappa$ is approximately
\be t_3\approx  \frac{2}{\alpha-\beta} W \left( \frac{ (\alpha-\beta)}{2\gamma} \sqrt{\frac{\mathrm{P}_{32}\mathrm{P}_{21}\chi_0}{2\kappa}}\right) .\label{eq:t3hardway} \ee
This expression is identical to (\ref{eq:linATex}).  }
{\color{black}
It is more tedious to derive estimates in this manner when there are more than one shortest path between nodes.  For example, suppose that node $1$ is connected to nodes $2$ and $3$ which are then both connected to node $4$.  Using the same assumptions as above we would then obtain that $j_4(t)$ should have an approximate solution of the form
\[ j_4(t)\approx \gamma \mathrm{P}_{42}\me^{(\alpha-\beta)t}\int_0^t \me^{-(\alpha-\beta)\tau}j_2(\tau)d\tau +\gamma \mathrm{P}_{43}\me^{(\alpha-\beta)t}\int_0^t \me^{-(\alpha-\beta)\tau}j_3(\tau)d\tau. \]
Using expressions for $j_2(t)$ and $j_3(t)$ we then would find 
\[ j_4(t)\approx \gamma^2 \left(\mathrm{P}_{42}\mathrm{P}_{21} +\mathrm{P}_{43}\mathrm{P}_{31}\right) \chi_0 \frac{t^2}{2} e^{(\alpha-\beta)t},\]
where we note that $\rho_4=\mathrm{P}_{42}\mathrm{P}_{21} +\mathrm{P}_{43}\mathrm{P}_{31}$ and we then find the same arrival time estimate as in (\ref{eq:linATW}).}

This process can then be continued and refined.  In terms of providing accurate arrival time estimates for (\ref{eq:main}) this method is cumbersome in comparison to the matrix exponential expansion performed in \cite{chen18}, however, it provides a different point of view to see how the arrival time estimates in (\ref{eq:linATex}) may be derived and will be used later in cases where the matrix exponential approach does generate accurate estimates. 

We conclude this section with three remarks. 

\begin{rmk} Suppose that in (\ref{eq:j3approx}) we had instead used the expression for $j_2(t)$ given in (\ref{eq:j2sol}).  Then our solution for $j_3(t)$ would read (approximately)
\[ j_3(t)\approx \gamma^2 \mathrm{P}_{32}\mathrm{P}_{21} \chi_0 \frac{t^2}{2} \me^{(\alpha-\beta) t},  \]
and the arrival time estimate would be exactly as that derived from (\ref{eq:themainlin}) despite the fact that $j_1(t)$ is, unrealistically, assumed to grow exponentially on the entire time interval $0<t<t_3$.    
\end{rmk}

\begin{rmk}  The arrival time estimates in (\ref{eq:linAT}) are observed in numerical simulations to be good predictors for arrival times in the nonlinear model if both $\chi_0$ and $\kappa$ are small (again in the limit as $\gamma\to 0$).  In light of the discussion above, we see that $\chi_0$ small is required so that $j_1(t)\approx \chi_0 \me^{(\alpha-\beta)t}$ is accurate while $\kappa$ small is needed so that the threshold is crossed when $j_n(t)$ is small and the approximation in (\ref{eq:j2approx}) is valid.
\end{rmk}
{\color{black}
\begin{rmk} Suppose that the local dynamics in (\ref{eq:main}) are changed to be SIS type dynamics where recovered individuals become susceptible at some rate $\gamma r_n$.  Linearizing at the disease free state, it once again turns out the the infected dynamics decouple and are described by (\ref{eq:linsol}).  Therefore, the linear arrival time estimates are exactly the same as in (\ref{eq:main}).  Numerical simulations of the SIS model show that the linear arrival times remain good estimates for the nonlinear arrival times in this model. 
\end{rmk}

}

\section{Arrival times for a SEIR model} \label{sec:SEIR}
The local dynamics in (\ref{eq:main}) are described by the simple SIR model.  We now demonstrate how to extend the arrival time estimates for other types of disease models.  For example, many disease models incorporate an exposed population that accounts for the latency in infection once an individual becomes infected with a disease.  The generalization of (\ref{eq:main}) to this case is 
\begin{eqnarray}
\partial_{t}s_{n}&=&-\alpha s_{n}j_{n}+ \gamma\sum_{m\neq n}P_{nm}(s_{m}-s_{n}) \nonumber \\
\partial_{t}e_{n}&=&\alpha s_{n}j_{n}-\sigma e_n +\gamma\sum_{m\neq n}P_{nm}(e_{m}-e_{n}) \nonumber \\
\partial_{t}j_{n}&=&\sigma e_n -\beta j_{n}+ \gamma\sum_{m\neq n}P_{nm}(j_{m}-j_{n}) \nonumber \\
\partial_{t}r_{n}&=&\beta j_{n}+ \gamma\sum_{m\neq n}P_{nm}(r_{m}-r_{n}) 
.\label{eq:seir} \end{eqnarray}
We demonstrate how to derive arrival time estimates in this case.  First, linearize about the disease free state $(1,0,0,0)^T$ to obtain (neglecting the recovered population once again) 
\begin{eqnarray*}
 \sv_t &=& -\alpha \jv +\gamma \left( \mathrm{P}-\mathrm{I}\right) \sv  \\
 \ev_t &=& -\sigma \ev + \alpha \jv +\gamma \left( \mathrm{P}-\mathrm{I}\right) \ev \\
\jv_t &=&  \sigma \ev-\beta \jv +\gamma \left( \mathrm{P}-\mathrm{I}\right) \jv .
\end{eqnarray*}
Note that the $\ev$-$\jv$ sub-system decouples.  Write this sub-system abstractly as 
\be \vh_t=\tilde{\mathrm{A}} \vh+\gamma \left(\tilde{\mathrm{P}}-\mathrm{I}\right)\vh, \label{eq:H} \ee
where $\vh=\left(e_1,j_1,e_2,j_2,\dots,e_N,j_N\right)^T$, the {\color{black} $2N\times 2N$ } matrix $\tilde{\mathrm{P}} = \mathrm{P} \otimes \mathrm{I}_2$, and  the {\color{black} $2N\times 2N$ } matrix $  \tilde{\mathrm{A}} = \mathrm{I}_N \otimes \mathrm{A}$ with
\begin{equation*}
    \mathrm{A} = \left(\begin{array}{cc}
         -\sigma & \alpha \\ \sigma & -\beta\end{array}\right) .
\end{equation*}
The matrix $\mathrm{A}$ is the local linearization of the reaction terms for (\ref{eq:seir}) at a fixed node.  Let 
\[ \lambda_\pm(\alpha,\sigma,\beta)=\frac{-(\beta + \sigma + 2\gamma) + \sqrt{(\beta - \sigma)^2 + 4\sigma\alpha}}{2},\] 
 be the two eigenvalues of $\mathrm{A}$ and note that since $\mathrm{det}(\mathrm{A})=-\sigma(\alpha-\beta)$ then if  $\alpha-\beta>0$ we have $\lambda_+>0>\lambda_-$ and the disease free state is unstable. Note that the instability threshold for the SEIR model is identical to that of the SIR model. The matrix $\mathrm{A}$ is diagonalizable. Let $\mathrm{A}=\mathrm{Q} \mathrm{D} \mathrm{Q}^{-1}$, with 
\[ \mathrm{D} = \left(\begin{array}{cc}
         \lambda_+ & 0 \\ 0 & \lambda_- \end{array}\right) , \quad \mathrm{Q} = \left(\begin{array}{cc}
         \Gamma_+ & \Gamma_- \\ 1 & 1 \end{array}\right),\]
 where  $\Gamma_\pm(\alpha,\sigma,\beta)  =  \frac{\beta - \sigma \pm \sqrt{(\beta - \sigma)^2 + 4\sigma\alpha}}{2\sigma}$.  Equation (\ref{eq:H}) can be solved using the matrix exponential as 
\[ \vh(t)=\me^{(\tilde{\mathrm{A}} + \gamma(\tilde{\mathrm{P}} - \mathrm{I}))t}\vh_0. \]
Key to the derivation of the arrival time estimate in the SIR model is the ability to separate the homogeneous growth due to the instability from the diffusion due to the coupling matrix $\mathrm{P}$.  Such a decomposition is possible here since the  matrices $\tilde{\mathrm{A}}$ and $\tilde{\mathrm{P}}$ commute which we verify using properties of the Kronecker product,
\begin{align}
    \tilde{\mathrm{A}}\tilde{\mathrm{P}} &= (\mathrm{I}_N \otimes A)(\mathrm{P} \otimes \mathrm{I}_2) = (\mathrm{I}_N \mathrm{P}) \otimes (A \mathrm{I}_2)  \nonumber \\ \nonumber
    &= (\mathrm{P} \mathrm{I}_N) \otimes (\mathrm{I}_2 A) = (\mathrm{P} \otimes \mathrm{I}_2)(\mathrm{I}_N \otimes A) = \tilde{\mathrm{P}}\tilde{\mathrm{A}}.
\end{align}
{\color{black} Since the matricies commute we can therefore write the solution }
\be \vh(t)=\me^{(\tilde{\mathrm{A}}-\gamma\mathrm{I})t}\me^{ \gamma\tilde{\mathrm{P}}t}\vh_0, \label{eq:seirlinsol} \ee
and expand the matrix exponentials as 
\[
    \me^{(\mathrm{I}_N \otimes \mathrm{A})t} = \sum_{j = 0}^{\infty} \frac{t^j(\mathrm{I}_N \otimes \mathrm{A})^j}{j!} = \sum_{j = 0}^{\infty}\frac{\mathrm{I}_N \otimes \mathrm{A}^j}{j!}t^j, \quad 
    \me^{(\mathrm{P} \otimes \mathrm{I}_2)t} = \sum_{k = 0}^{\infty} \frac{t^k (\mathrm{P} \otimes \mathrm{I}_2)^k}{k!} = \sum_{k = 0}^{\infty}\frac{\mathrm{P}^k \otimes \mathrm{I}_2}{k!}t^k .
\]
To calculate arrival times for a disease propagating from city $n$ to city $m$, we specify that at time zero we have some proportion, $\chi_0$, of the infected population in city $n$ and calculate when the infected population exceeds some threshold $\kappa$ at city $m$.  Thus, $\vh_0= \chi_0 \boldsymbol{\delta}_n \otimes \tilde{\boldsymbol{\delta}}_2$, where $\boldsymbol{\delta}_n$ denotes the standard Euclidean basis vector in $\mathbb{R}^N$ while $\tilde{\boldsymbol{\delta}}_j$ is the same for $\mathbb{R}^2$.   This leads to the following equation to determine the arrival times, which we simplify using properties of the Kronecker product,

\begin{eqnarray*} \kappa&=&  \chi_0 \left( \boldsymbol{\delta}_m \otimes\tilde{\boldsymbol{\delta}}_2 \right)^T \left( \mathrm{I}_N \otimes \mathrm{Q} \me^{\mathrm{D}\tau_{mn}} \mathrm{Q}^{-1}\right) \left( \sum_{k = 0}^{\infty}\frac{\gamma^k \mathrm{P}^k \otimes \mathrm{I}_2}{k!}\tau_{mn}^k  \left( \boldsymbol{\delta}_n \otimes \tilde{\boldsymbol{\delta}}_2  \right) \right) \\
&=&  \chi_0 \left( \boldsymbol{\delta}_m^T \otimes \tilde{\boldsymbol{\delta}}_2^T \mathrm{Q} \me^{\mathrm{D}\tau_{mn}} \mathrm{Q}^{-1}\right) \left( \sum_{k = 0}^{\infty}\frac{\gamma^k \mathrm{P}^k \boldsymbol{\delta}_n \otimes \tilde{\boldsymbol{\delta}}_2}{k!}\tau_{mn}^k \right) \\
&=&  \chi_0\left( \sum_{k=0}^\infty\gamma^k  \frac{\boldsymbol{\delta}_m^T \mathrm{P}^k \boldsymbol{\delta}_n}{k!}\tau_{mn}^k \right) \otimes \left( \tilde{\boldsymbol{\delta}}_2^T \mathrm{Q} \me^{\mathrm{D}\tau_{mn}} \mathrm{Q}^{-1} \tilde{\boldsymbol{\delta}}_2 \right) .
\end{eqnarray*}
Since both terms in parenthesis are scalar the Kronecker product in the last line is actually just a multiplication.  Assuming again that the leading order term in $\gamma$ will dominate we can neglect all terms in the sum aside from the one where $k=d_m$.  For the terms on the right, we simplify to
\[ \tilde{\boldsymbol{\delta}}_2^T \mathrm{Q} \me^{\mathrm{D}\tau_{mn}} \mathrm{Q}^{-1} \tilde{\boldsymbol{\delta}}_2= \frac{1}{\Gamma_--\Gamma_+} \left(\Gamma_- \me^{\lambda_+ \tau_{mn}} - \Gamma_+ \me^{\lambda_- \tau_{mn}} \right). \]
We neglect the exponential involving $\lambda_-$ since $\lambda_-<0$ and obtain 
\[ \kappa = \chi_0 \frac{\gamma^{d_m} \rho_m \tau_{mn}^{d_m}}{d_m!}\left( \frac{\Gamma_-}{\Gamma_--\Gamma_+}\right) \me^{\lambda_+ \tau_{mn}}.\] 
As was the case for the SIR model, this equation can be solved using the Lambert-W function and we obtain the estimate

\be \tau_{mn}=\frac{d_m}{\lambda_+(\alpha,\sigma,\beta)} W\left(\frac{1}{\gamma}  \frac{\lambda_+(\alpha,\sigma,\beta) }{ d_m} \left(\frac{\kappa d_m! (\Gamma_-(\alpha,\sigma,\beta)-\Gamma_+(\alpha,\sigma,\beta))}{\chi_0 \rho_m \Gamma_-(\alpha,\sigma,\beta)}\right)^{1/d_m} \right)+o(1). \label{eq:ATSEIR} \ee
{\color{black} Recall that as $\sigma\to \infty$ the period of time that individuals spend in the exposed phase tends to zero and we anticipate that the arrival times for the SEIR model should approach those for the SIR model in this limit.  Indeed, we observe that as $\sigma\to \infty$, $\Gamma_\pm\to -\frac{1}{2}\pm \frac{1}{2}$ and $\lambda_+\to \alpha-\beta$  and the arrival time estimate (\ref{eq:ATSEIR}) converges to the estimate for the SIR model, see (\ref{eq:linATW}).   Comparisons between this arrival time and those in direct numerical simulations of (\ref{eq:seir}) are presented in Figure~\ref{fig:SEIR}.  }

\begin{figure}
    \centering
     \subfigure{\includegraphics[width=0.3\textwidth]{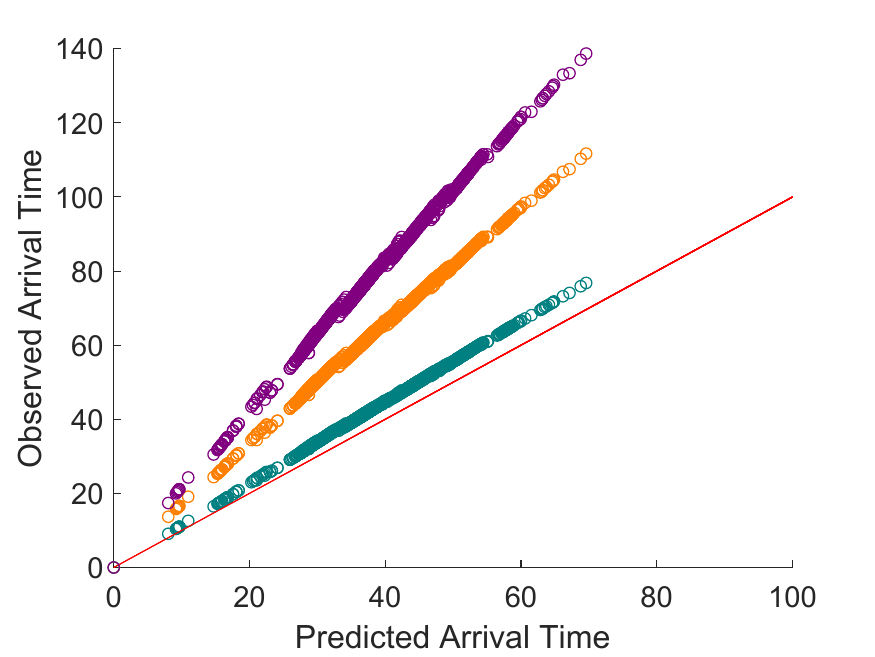}}
 \subfigure{\includegraphics[width=0.3\textwidth]{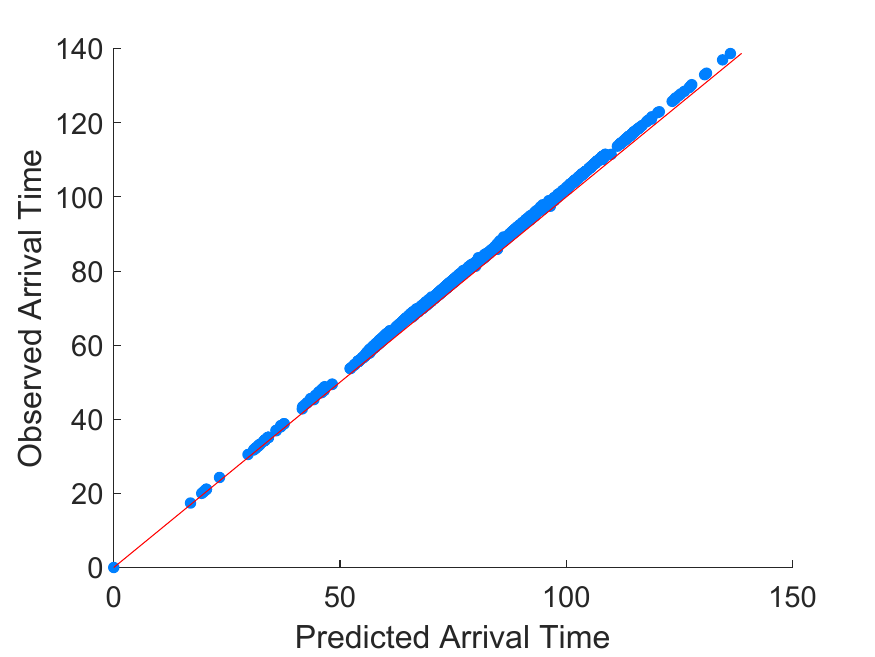}}
  \subfigure{\includegraphics[width=0.3\textwidth]{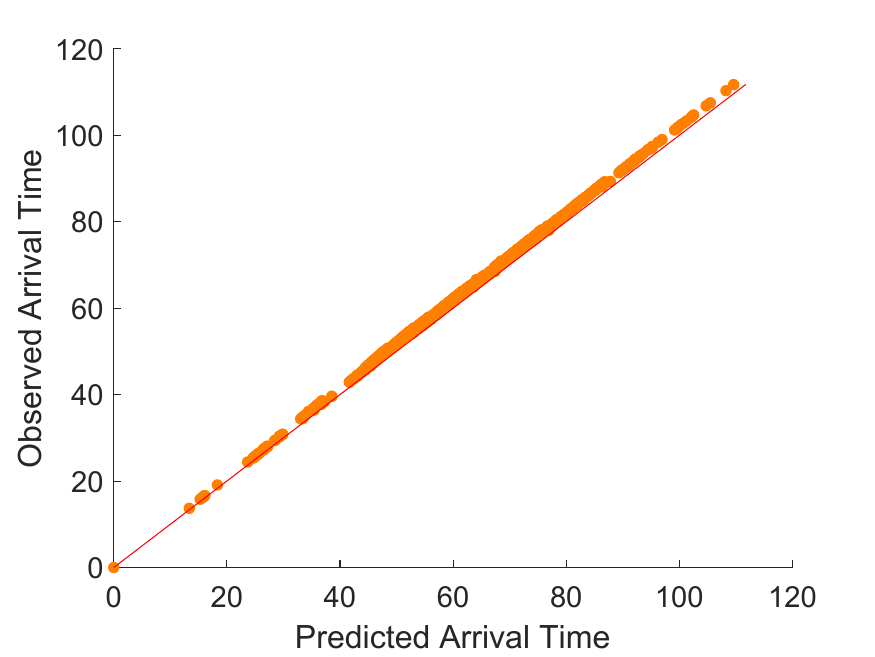}}
   \caption{Arrival times for the SEIR model (\ref{eq:seir}) versus predictions.  In the left panel we show arrival times observed in numerical simulations versus predicted arrival times based upon the arrival time estimate for the SIR model; see (\ref{eq:linATW}).  Three simulations are performed with $\sigma=0.5$, $\sigma=1.0$ and $\sigma=10$.  As anticipated the arrival time of the disease is delayed by the incorporation of an exposed phase.  For $\sigma=10$, individuals reside in the exposed phase for a short amount of time and the arrival times for the SEIR model are close to those of the SIR model.  In the right two panels, we compare observed arrival times in the SEIR model to the prediction (\ref{eq:ATSEIR}) for $\sigma=0.5$ (middle panel) and $\sigma=1.0$ (right panel).  Here $\alpha=1$, $\beta=0.25$ and $\gamma=0.001$.   }
    \label{fig:SEIR}
\end{figure}

{\color{black}
\begin{rmk}  Arrival times estimates analogous to formula (\ref{eq:ATSEIR}) can be obtained for other compartmental models as well.  Let $A$ denote the linearization of the local model near the disease free equilibrium point of an $\ell$ component disease model.  Then the linearization at the unstable state in the network system can be expressed as 
\[ \vh_t=\tilde{\mathrm{A}} \vh+\gamma \left(\tilde{\mathrm{P}}-\mathrm{I}\right)\vh, \]
for $\vh\in\mathbb{R}^{N\ell}$.  Assume that the components of $\vh$ are ordered so that the infected proportion is expressed first. Then the solution of the linear equation can be written as 
\[ \vh(t)=\me^{(\tilde{\mathrm{A}}-\gamma\mathrm{I})t}\me^{ \gamma\tilde{\mathrm{P}}t}\vh_0, \]
where $\tilde{\mathrm{P}} = \mathrm{P} \otimes \mathrm{I}_\ell$, and  $  \tilde{\mathrm{A}} = \mathrm{I}_N \otimes \mathrm{A}$
Suppose that $\mathrm{A}$ is diagonalizable with maximal eigenvalue $\lambda_1(\mathrm{A})$ and $\mathrm{A}=\mathrm{Q}\mathrm{D}\mathrm{Q}^{-1}$ where $\mathrm{Q}=\left( \mathbf{q}_1 \ \mathbf{q}_2 \ \dots \ \mathbf{q}_\ell\right)$ with $\mathbf{q}_j$ the eigenvectors of $\mathrm{A}$.  Then  $\lambda_1$ is the upper left entry of $\mathrm{D}$. Let $\tilde{\mathrm{D}}=\mathrm{diag}(\lambda_1, 0,\dots,0)$.  Then we can estimate arrival times by solving
\[ \kappa = \chi_0 \frac{\gamma^{d_m} \rho_m \tau_{mn}^{d_m}}{d_m!} \Xi(\mathrm{A}) \me^{\lambda_+ \tau_{mn}},\] 
where the constant $\Xi(\mathrm{A})$ comes from projecting the initial condition of the local dynamics onto the leading eigenvector and is defined as
\[ \Xi(\mathrm{A})=\tilde{\boldsymbol{\delta}}_1^T \mathrm{Q} \tilde{\mathrm{D}}\mathrm{Q}^{-1}\tilde{\boldsymbol{\delta}}_1= \frac{q_{11}}{\mathrm{det}(\mathrm{Q})} \mathrm{det}\left[ \boldsymbol{\delta}_1 \ \mathbf{q}_2  \ \mathbf{q}_3 \ \dots \ \mathbf{q}_\ell \right]. \]
This leads to the arrival time estimate 
\be  \tau_{mn}=\frac{d_m}{\lambda_1(\mathrm{A})} W\left(\frac{1}{\gamma}  \frac{\lambda_1(\mathrm{A}) }{ d_m} \left(\frac{\kappa d_m! \Xi(\mathrm{A}) }{\chi_0 \rho_m }\right)^{1/d_m} \right)+o(1).  \label{eq:genATformula} \ee
\end{rmk}
}

\section{Pushed fronts: faster invasion speeds due to nonlinearities} \label{sec:pushed} 

{\color{black} 
Not all invasion fronts are linearly determined.   For the SIR model considered in (\ref{eq:main}) the nonlinearity suppresses growth and the maximal growth rate of the infection occurs when when the infected population is small.  In this section, we demonstrate that nonlinearities which amplify growth can lead to faster-than-linear arrival times.  This phenomena is well known in the PDE setting where the resulting fronts are referred to as {\em pushed}; see for example \cite{hadeler75,vansaarloos03}.

Consider the following meta-population model,
\begin{eqnarray}
\partial_{t}s_{n}&=&-\alpha s_{n}j_{n}-\rho s_nj_n^2+ \gamma\sum_{m\neq n}P_{nm}(s_{m}-s_{n}) \nonumber \\
\partial_{t}j_{n}&=&\alpha s_{n}j_{n}+\rho s_nj_n^2-\beta j_{n}+ \gamma\sum_{m\neq n}P_{nm}(j_{m}-j_{n}) 
.\label{eq:pushed} \end{eqnarray}
The only difference between this system and (\ref{eq:main}) is the additional infection term $\rho s_n j_n^2$.   This system is motivated by recent work in \cite{iacopini19} where the role of higher order interactions in social epidemics is studied. Recall that the quadratic terms $\alpha s_nj_n$ represent infections occurring due to interactions between infected and susceptible individuals.  The cubic term  $\rho s_n j_n^2$ represents infections due to group (we consider only groups of size three for simplicity) interactions and expresses the higher probability of a susceptible individual adopting a new social norm if all the other members of one of their social groups has already adopted that norm.   We emphasize that the model in \cite{iacopini19} is an agent-based stochastic model without spatial structure and refer the reader to \cite{iacopini19} for more details.

We are interested in how these higher-order interactions affect arrival times.  Based upon our analysis of the SIR and SEIR models the natural starting point is to compute linearly determined arrival times.  In fact, the linearization of (\ref{eq:pushed}) near the unstable disease free state is equivalent to that of (\ref{eq:main}) and therefore the linear arrival time estimates for this system are also identical.  However, numerical simulations reveal faster invasion speeds; see Figure~\ref{fig:pushed}.  We proceed to explain and predict this faster invasion speed starting first with a discussion of the local dynamics of (\ref{eq:pushed}).

\subsection{The local dynamics }
To obtain modified arrival times estimates using the approach presented in Section~\ref{sec:cascade} we need an estimate for the local dynamics of (\ref{eq:pushed}) at a fixed city in the absence of diffusion.  In this section, we obtain an approximation for these dynamics in the limit as $\rho\to\infty$. This corresponds to a regime where infections via group interactions dominates those stemming from pairwise interactions. 

Consider the local dynamics of (\ref{eq:pushed}), 
\begin{eqnarray}
S'&=& -\alpha SI -\rho SI^2  \nonumber \\
I'&=& \alpha SI +\rho SI^2 -\beta I. \label{eq:SIRsimplicial}
\end{eqnarray}
We desire estimates on the solution of (\ref{eq:SIRsimplicial}) for initial conditions starting near the disease free steady state $(S,I)=(1,0)$.  We will consider the case when $\rho\gg 1$ so that we can view (\ref{eq:SIRsimplicial}) as a singularly perturbed system.  Let $\e=\frac{1}{\rho} \ll 1$.  After transformation of the independent variable by $\tau=\frac{t}{\e}$ we obtain the following system of equations 
\begin{eqnarray}
\frac{dS}{d\tau} &=& -SI^2-\e \alpha SI   \nonumber \\
\frac{dI}{d\tau} &=& SI^2 +\e \alpha SI -\e \beta I. \label{eq:T1}
\end{eqnarray}
Setting $\e=0$ we obtain the  so-called reduced fast equation,
\begin{eqnarray}
\frac{dS}{d\tau} &=& -SI^2 \nonumber \\
\frac{dI}{d\tau} &=& SI^2. \label{eq:T1fast}
\end{eqnarray}
This reduced equation is, to leading order, the same as the system of equations analyzed in \cite{gucwa09} and so we follow their analysis.  System (\ref{eq:T1fast}) has two lines of equilibria: in the language of Geometric Singular Perturbation Theory these are called slow manifolds --  $\mathcal{M}_I=\{ (S,I) \ | \ S=0 \ \}$ and $\mathcal{M}_S=\{ (S,I) \ | \ I=0 \ \}$; see for example \cite{jones95}.  The two manifolds intersect at the origin.  For $I>0$, the manifold $\mathcal{M}_I$ is normally hyperbolic whereas the manifold $\mathcal{M}_S$ lacks normal hyperbolicity.   

Let $W=S+I$. Then for (\ref{eq:T1fast}), $W(\tau)$ is constant to leading order while  
\[
\frac{dI}{d\tau} = (W-I)  I^2. \]
For $\e$ small and away from $\mathcal{M}_S$ we therefore have that $W(\tau)$ is constant to leading order in $\e$ while $I(\tau)$ increases from zero to $W$.  This provides a leading order fast connection between the slow manifolds $\mathcal{M}_S$ and $\mathcal{M}_I$.  

In order to use the local solution to estimate arrival times in (\ref{eq:pushed}) we need some basic estimates on the form of the solution starting near $(1,0)$ for small $\e$.  We will consider the initial condition $I(0)=\kappa$ and $S(0)=1-\kappa$ with $\frac{\kappa}{\e}$ sufficiently small so that $\kappa<\e$.  We must then follow this initial condition in the slow time scales until $I(t)$ exceeds some threshold $\eta$ and the fast dynamics prescribed by (\ref{eq:T1fast}) take over and the solution quickly converges to the slow manifold $\mathcal{M}_I$.  Once it nears $\mathcal{M}_I$, then the solution relaxes exponentially to the origin since $I'\approx -\beta I$ there.  

Since $\mathcal{M}_S$ is not normally hyperbolic we can not directly use the linearization to estimate the solution before the transition time $\Omega$.  This lack of normal hyperbolicity can be traced to the fact that as $I\to0$ the dominant term on the right side of (\ref{eq:T1}) shifts from $SI^2$, which is formally $\O(1)$ in $\e$ to $\e I (\alpha S-\beta)$ which is formally $\O(\e)$.  

In appendix~\ref{sec:appendix} we mimic the geometric desingularization approach of \cite{gucwa09} to obtain estimates on $\Omega$.  We will obtain an approximation for the transition time 
\be \Omega\approx\frac{1}{\alpha-\beta}\log\left(\frac{\e (\alpha-\beta)}{\kappa}\right), \label{eq:Omega} \ee
and the solution to the local dynamics as 
\be I(t)\approx \left\{ \begin{array}{cc} \kappa \me^{(\alpha-\beta)t} & t<\Omega \\ \me^{-\beta (t-\Omega) } & t\geq \Omega. \end{array} \right. \label{eq:Ipushed} \ee
Solutions of the system (\ref{eq:SIRsimplicial}) are shown in Figure~\ref{fig:pushedlocal}.

\begin{figure}
    \centering
     \subfigure{\includegraphics[width=0.33\textwidth]{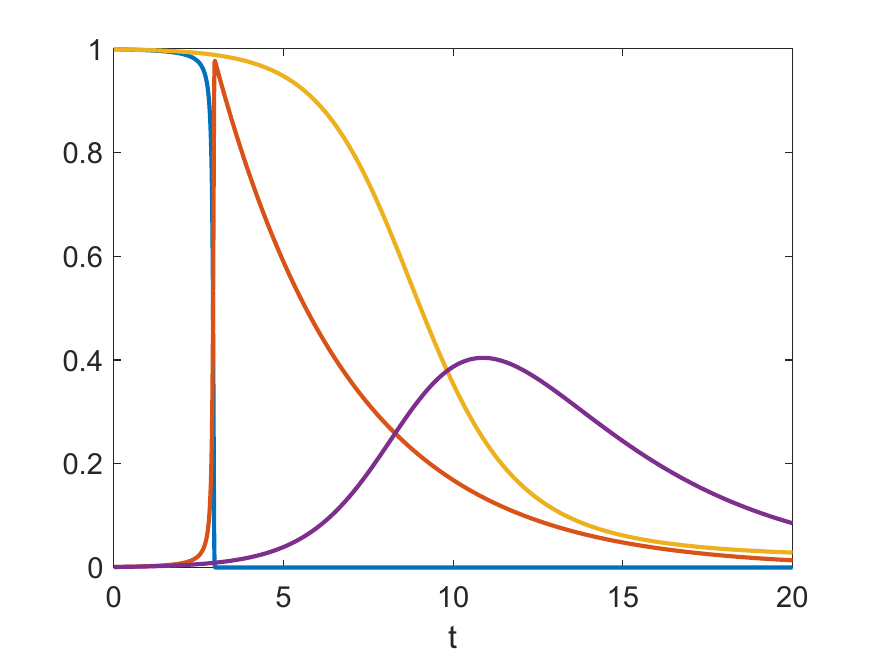}}
 \subfigure{\includegraphics[width=0.33\textwidth]{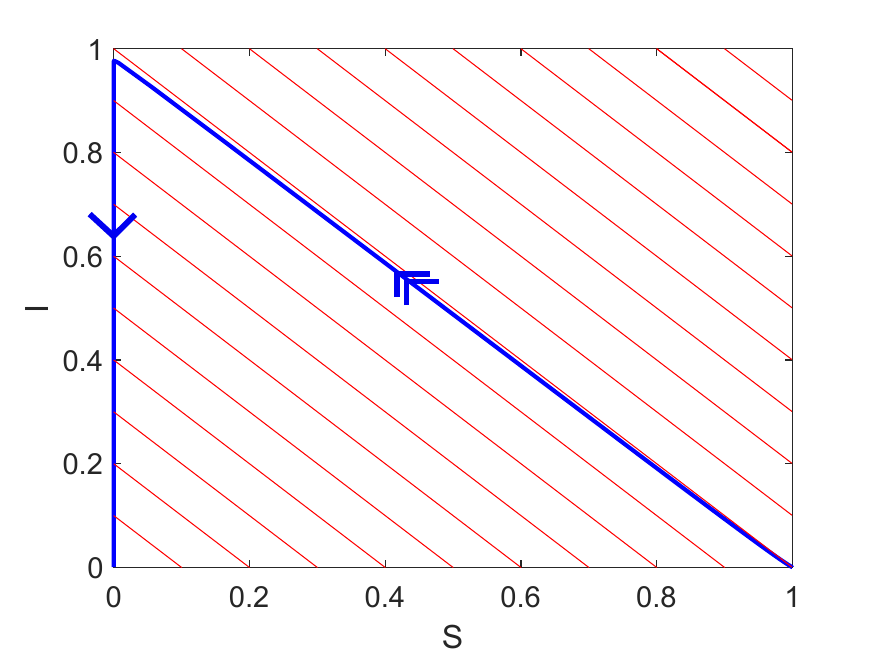}}
   \caption{The local dynamics for the system of equations in (\ref{eq:SIRsimplicial}).  On the left are numerically computed solution trajectories starting from the initial conditiond $I(0)=0.001$ and $S(0)=0.999$ for $\rho=100$ and $\rho=0$ (the standard SIR model) with $\alpha=1.50$ and $\beta=0.25$.    The red and blue curves are the infected and susceptible proportions respectively with $\rho=100$ while the purple and yellow are the infected and susceptible proportions when $\rho=0$.  Note that when $\rho=100$ the infected proportion has fast transition around $t=2.80$ with a theoretical estimate of $\Omega=2.6865$.  On the right, the solution with $\rho=100$ is plotted in $S-I$ phase space.  The red lines depict the invariant fast transition curves that connect the two slow manifolds.   }
    \label{fig:pushedlocal}
\end{figure}

\subsection{Arrival time estimates}  
We now turn our attention to making estimates of nonlinear arrival times using the approach outlined in Section~\ref{sec:cascade}.   Assume that the epidemic originates at node $n=1$.   For simplicity we assume that the initial infected proportion at city $1$ is $\kappa$ as in (\ref{eq:Ipushed}).  We then wish to estimate the nonlinear arrival times $t_m$ defined by the condition that $j_m(t_m)=\kappa$.  In fact, we will obtain the estimate
\be  t_m\approx \frac{d_m}{\alpha-\beta} \log\left(\frac{\kappa\alpha}{\gamma(\rho_m)^{1/ d_m} }\right)+d_m \Omega, \label{eq:tm} \ee
where we recall that $d_m$ is the minimal number of flights connecting city $m$ to the origin city (the graph distance),  $\rho_m$ is the random walk probability of moving between city $m$ and the origin city in exactly $d_m$ stops and $\Omega$ is the local transition time obtained (\ref{eq:Omega}).  

To verify this we first consider the city $n=2$ which we suppose is connected to origin node.  We  approximate the evolution of the infected population at node two by the equation
\[ \frac{dj_2}{dt}\approx (\alpha-\beta) j_2 +\gamma \mathrm{P}_{21} j_1(t). \]
Supposing that $j_1(t)$ evolves according to (\ref{eq:Ipushed})  leads to an approximate expression for $j_2(t)$,
\[ j_2(t)\approx \gamma \mathrm{P}_{21} \me^{(\alpha-\beta)t}\left( \int_0^\Omega \kappa d\tau +\int_\Omega^t \me^{-\alpha \tau} \me^{\beta\Omega} d\tau\right).  \]
The contribution from the second integral dominates and we estimate the arrival time by setting $j_2(t)=\kappa$ while ignoring lower order terms yields the equation
\[ \kappa = \gamma \frac{\mathrm{P}_{21}}{\alpha} \me^{(\alpha-\beta) (t-\Omega)}, \]
from which we estimate the arrival time $t_2$ by 
\[ t_2=\frac{1}{\alpha-\beta} \log \left(\frac{ \kappa \alpha}{\gamma \mathrm{P}_{21}}  \right)+\Omega, \]
which agrees with (\ref{eq:tm}).  
Extrapolating, we can consider the evolution at an arbitrary node $m$ where the evolution of the infected population is approximately governed by the following differential equation 
\[ \frac{dj_m}{dt} \approx (\alpha-\beta)j_m +\gamma\sum_{k: d_k=d_m-1} \mathrm{P}_{mk} j_k(t). \]
The sum represents the coupling to cities which are closer to the origin city.  For each $j_k$ we substitute
\[ j_k(t)\approx \left\{ \begin{array}{cc} \kappa \me^{(\alpha-\beta)(t-t_k)} & t-t_k<\Omega \\
 e^{-\beta(t-t_k-\Omega)} & t-t_k\geq \Omega\end{array}\right. \]
As we did for the node $j_2$, we approximate this solution as
\[ j_m(t)\approx \gamma  \me^{(\alpha-\beta) t} \sum_{k: d_k=d_m-1} \int_{t_k+\Omega}^t\mathrm{P}_{mk} \me^{-\alpha\tau} \me^{\beta (t_k+\Omega)} d\tau, \]
after which integrating and neglecting the upper bound of integration we obtain an arrival time estimate by solving 
\be \kappa =  \gamma \me^{(\alpha-\beta)t}   \sum_{k: d_k=d_m-1}\mathrm{P}_{mk} \me^{-(\alpha-\beta) (t_k+\Omega)} \label{eq:almosthere} . \ee
Using 
\[ t_k=\frac{d_k}{\alpha-\beta} \log\left(\frac{\kappa\alpha}{\gamma(\rho_k)^{1/ d_k} }\right)+d_k \Omega,\]
then (\ref{eq:almosthere}) becomes
\[ \kappa =\gamma\me^{(\alpha-\beta)(t-d_m\Omega)}   \sum_{k: d_k=d_m-1}  \frac{\gamma^{d_k}\rho_k}{\kappa^{d_k} \alpha^{d_k}}\frac{\mathrm{P}_{mk}}{\alpha}. \]
Since $\rho_m=\sum_{k: d_k=d_m-1} \rho_k\mathrm{P}_{mk}$ this is equivalent to 
\[ \frac{(\alpha \kappa)^{d_m}}{\gamma^{d_m} \rho_m}= \me^{(\alpha-\beta)(t_m-d_m\Omega)}, \]
and solving for $t_m$ we obtain the expression in (\ref{eq:tm}).

In Figure~\ref{fig:pushed}, we show comparisons between this arrival time estimate and those observed in numerical simulations.   }

\begin{figure}
    \centering
     \subfigure{\includegraphics[width=0.3\textwidth]{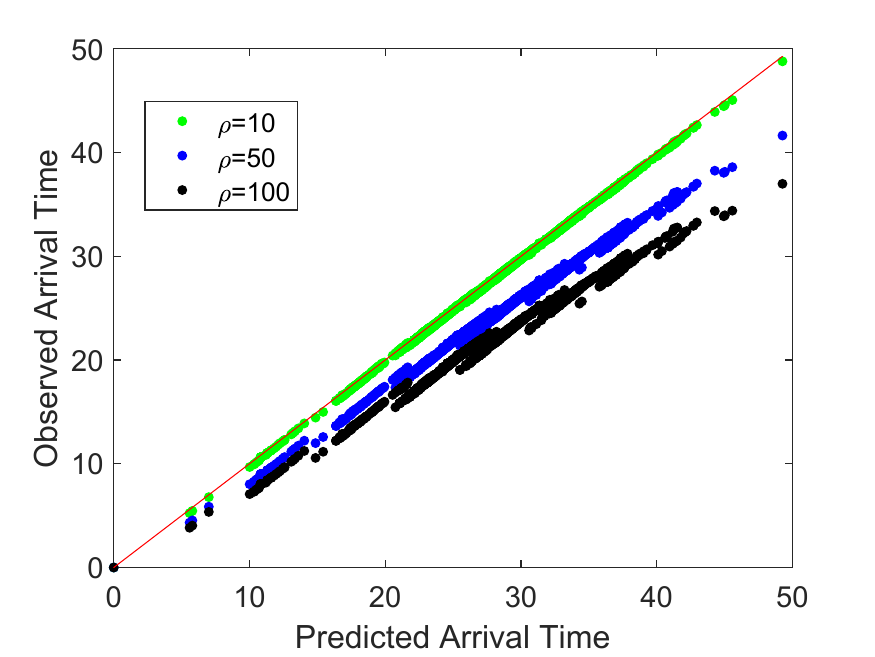}}
 \subfigure{\includegraphics[width=0.3\textwidth]{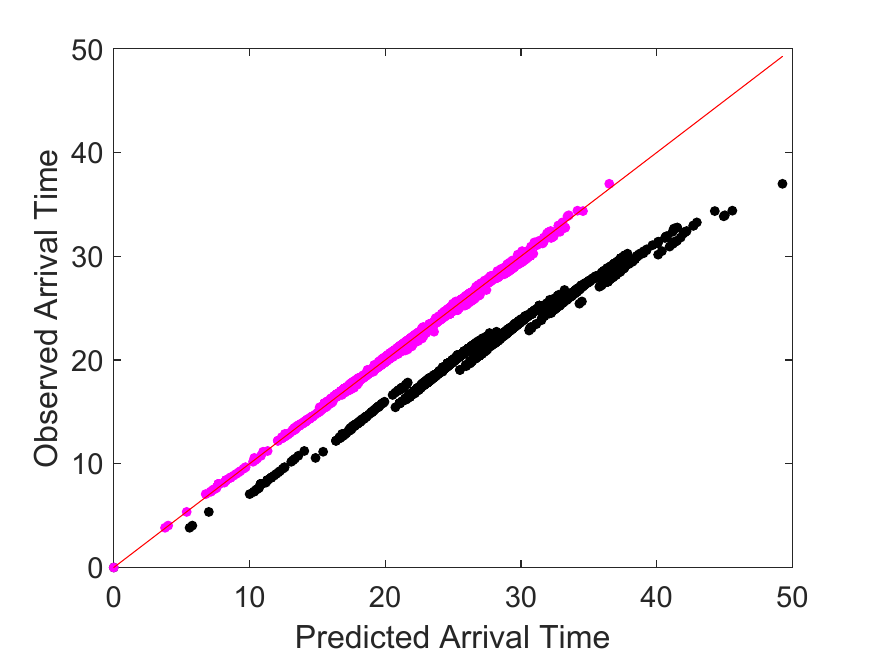}}
  \subfigure{\includegraphics[width=0.3\textwidth]{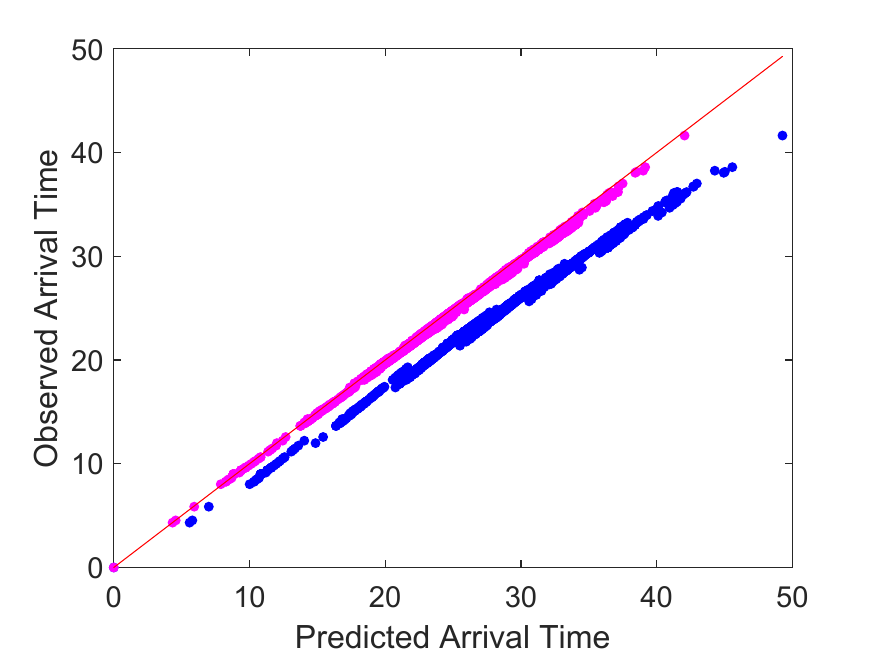}}
   \caption{Epidemic arrival times for (\ref{eq:pushed}) are plotted against various predictions.  On the left, we plot numerically observed arrival times in (\ref{eq:pushed}) for $\rho=10,50,100$ versus the linear arrival time estimate (\ref{eq:linATW}).  Note that large values of $\rho$ lead to faster invasion speeds.  In the other two panels we plot in magenta observed arrival times against our nonlinear prediction in  (\ref{eq:tm}) for $\rho=100$ (middle panel, with original data for comparison) and $\rho=50$ (right panel, with original data for comparison).  For all simulations, $\alpha=1.5$, $\beta=0.25$ and $\gamma=0.001$.   }
    \label{fig:pushed}
\end{figure}

\section{Inhomogeneous infection rates speed up average arrival times} \label{sec:inhomo}
System (\ref{eq:main}) assumes that local infection and recovery rates are uniform across all cities.  In this section, we consider how inhomogeneties in these rates affect arrival times by allowing the infection rate to vary by node.  {\color{black} We will suppose that the infection rate at each node is expressed as $\alpha+\omega_n$ where $\alpha$ is the mean infection rate and $\omega_n$ describes city by city variations from this mean.}  Local infection rates are expected to differ for a variety of factors and we point out that rather large differences are reasonable, for example, for diseases that exhibit seasonality where the infection rate may vary by hemisphere.  The question we will focus on is whether this inhomogeneity speeds up or slows down the invasion process as compared to the average.  We consider the system
\begin{eqnarray}
\partial_{t}s_{n}&=&-\alpha s_{n}j_{n}-\omega_n s_nj_n+ \gamma\sum_{m\neq n}P_{nm}(s_{m}-s_{n}) \nonumber \\
\partial_{t}j_{n}&=&\alpha s_{n}j_{n}+\omega_n s_nj_n -\beta j_{n}+ \gamma\sum_{m\neq n}P_{nm}(j_{m}-j_{n}),  \label{eq:inhomo} \end{eqnarray}
where $\sum_{n=1}^N \omega_n=0$ and $\alpha+\omega_n-\beta>0$ for all $n$.  

A similar argument as in Theorem~\ref{thm:super} shows that the linear arrival times once again place a lower bound on nonlinear arrival times.  However, in contrast to the SIR or SEIR models, in the inhomogeneous case the linear arrival times are no longer a reliable predictor for the nonlinear arrival times.  We make two observations.  First, if we write (\ref{eq:inhomo}) in vector form then due to the inhomogeneity of the reaction terms it is no longer the case that the reaction and migration matrices commute, so it is not possible to decompose the solution as in (\ref{eq:linsol}) or (\ref{eq:seirlinsol}).  More problematic is the fact that the linearized solution will be dominated by the largest eigenvalue, corresponding to the largest $\omega_n$,  and so the linear equation will asymptotically predict arrival times equivalent to the homogeneous case with infection rate equal to $\alpha+\max_n \omega_n$.  We refer the reader to Section~\ref{sec:cascade} to see why this unbounded growth does not degrade the arrival time estimate in the homogeneous case.

\begin{figure}
    \centering
     \subfigure{\includegraphics[width=0.3\textwidth]{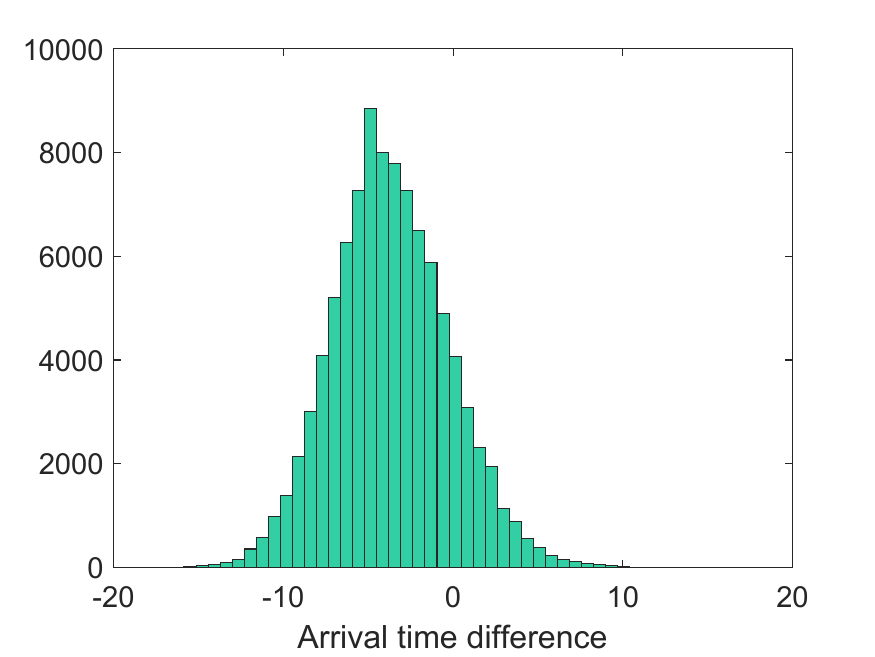}}
 \subfigure{\includegraphics[width=0.3\textwidth]{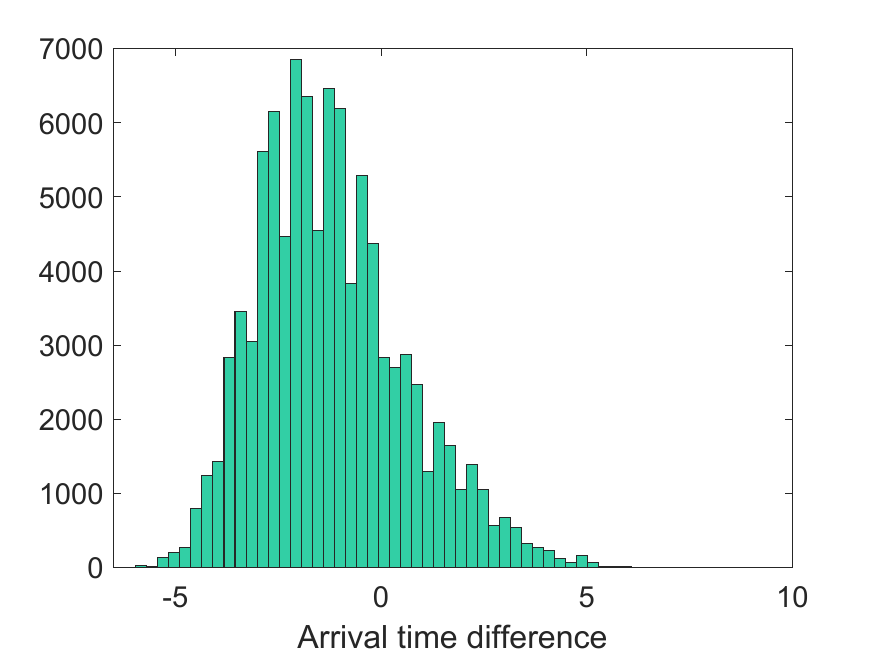}}
  \subfigure{\includegraphics[width=0.3\textwidth]{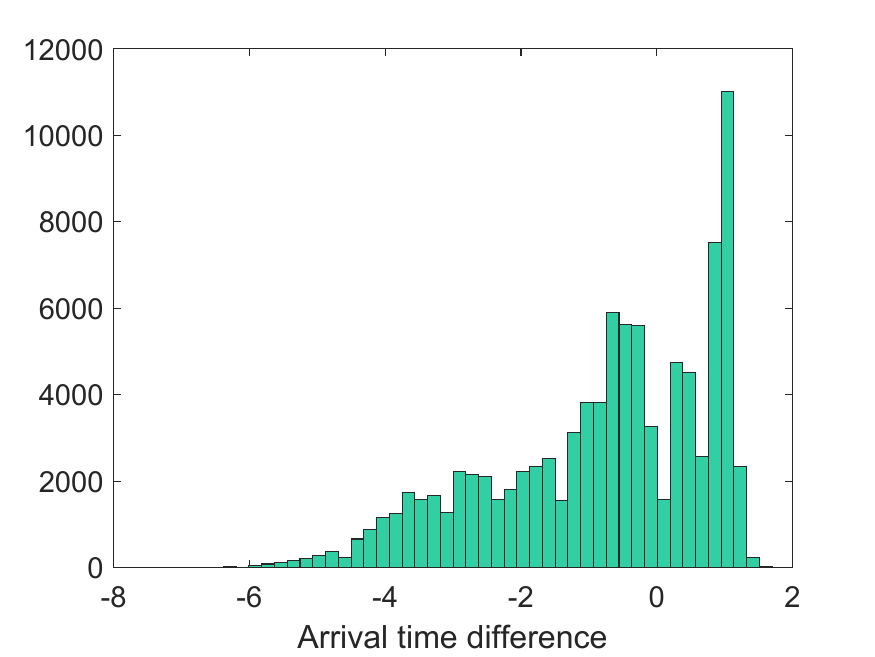}}
   \caption{Histograms showing the difference in arrival times between the inhomogeneous SIR model (\ref{eq:inhomo}) and the homogeneous model with constant infection rate equal to the mean of the inhomogeneous model.  Each figure represents observations over thirty different realizations of the random infection rates.  In each example $\alpha=1.0$, $\beta=0.25$ and $\gamma=0.001$.  On the left, $\alpha_n$ is drawn from a normal distribution, scaled by $0.2$ and then normalized to have zero mean.  On average, the epidemic arrives $3.66$ days earlier in the inhomogeneous model versus the homogeneous version.  In the middle panel, we randomly select half the nodes to have $\omega_n=0.2$ and the other half to have $\omega_n=-0.2$.  In this example, the arrival times are advanced by an average of $1.20$ days in the inhomogeneous versus homogeneous model.  On the right, we assign randomly one tenth of the nodes to have $\omega_n=0.18$ while the remaining nodes have $\omega_n=-0.02$.  Again, the inhomogeneous network has faster on average arrival times with a mean of $0.814$ days.       }
    \label{fig:inhomo}
\end{figure}

Numerical simulations suggest that arrival times in the inhomogeneous system are faster on average than arrival times in the homogeneous system.  These results are shown in Figure~\ref{fig:inhomo} for three different types inhomogeneities.  We argue that these faster arrival times are due to the following mechanism.  For the worldwide airline transportation network, most cities are connected by multiple shortest paths.  If the inhomogeneities are distributed randomly then it is likely that one of these shortest paths will connect the two cities along a route consisting entirely, or mostly, of cities with $\omega_n>0$.   Consulting (\ref{eq:linATex}) we expect this increase in  infection rate to decrease the arrival times at $\O(-\log(\gamma))$, whereas limiting the disease to spread along fewer of the possible shortest paths will decrease the random walk probability of traversing between the two cities.
  However, according to (\ref{eq:linATex}) this would only affect the arrival times at $\O(1)$.  We substantiate this point of view with some formal calculations as in Section~\ref{sec:cascade}.

For the nodes connected to the origin node, arrival time estimates can be derived as in Section~\ref{sec:cascade}.  Suppose that the disease originates at node $n=1$ and this node is connected to node $n=2$.  Let $\Gamma_n=\alpha+\omega_n-\beta$.  Then we approximate the dynamics of the infected proportion at node $2$ by 
\be j_2(t)\approx \gamma \mathrm{P}_{21} \chi_0 e^{\Gamma_2 t}\int_0^t e^{-\Gamma_2 \tau}e^{\Gamma_1 \tau }d\tau=\gamma \mathrm{P}_{21} \chi_0 e^{\Gamma_2 t}\left[ \frac{e^{(\omega_1-\omega_2)\tau}}{\omega_1-\omega_2}\right]_0^t. \label{eq:int} \ee
Setting this equal to the threshold value $\kappa$ we find two different arrival time estimates depending on whether $\omega_1>\omega_2$ or vice versa.  Let the arrival time $t_2$ be defined by $j_2(t_2)=\kappa$, then we get
\[ t_2\approx-\frac{1}{\Gamma_1}\log\left(\frac{1}{\gamma}\frac{\kappa (\omega_1-\omega_2)}{\chi_0 \mathrm{P}_{21}}\right), \quad \omega_1>\omega_2, \quad  t_2\approx-\frac{1}{\Gamma_2}\log\left(\frac{1}{\gamma}\frac{\kappa (\omega_2-\omega_1)}{\chi_0 \mathrm{P}_{21}}\right), \quad \omega_2>\omega_1. \]
{\color{black} In the case $\omega_1>\omega_2$ one can interpret the estimate as saying that the growth in infections at city $2$ is dominated by migration of infections from city $1$ where the local growth rate is larger.  In contrast, if $\omega_2>\omega_1$ then the growth of local infections at city $2$ dominates and the coupling to city $1$ is only required to transmit a few initial  infections to city $2$.} 
Both of these estimates rely on a gap between the $\omega_1$ and $\omega_2$ values so that one of the boundary terms in the integral in (\ref{eq:int}) can be ignored.  If these values are comparable then both terms need to be considered and the arrival time estimate will involve an approximation of the  Lambert-W function.  

The purpose of these informal calculations is to demonstrate that arrival times can be decreased by the disease passing through nodes with higher than average growth rates.  Now consider the grandchildren of the origin node.  These nodes are connected to the origin node through one or more children nodes.  For networks such as the worldwide airline network there are typically multiple such paths.  Thus, even if there is only a $1/2$ probability that the children nodes have higher than mean infection rates, there is a greater than even probability that there is a path with positive $\omega_n$ connecting the grandchild node to the origin.  This means that there exists a path over which the disease can spread faster leading to faster arrival times.  Numerical evidence for this is presented in Figure~\ref{fig:inhomoexplanation}.  {\color{black} Here we consider the worldwide airline network \cite{openflights} where each node has mean infection rate $\alpha=1.0$ and deviation $\omega_n=\pm 0.2$ selected uniformly at random}.  We then plot arrival times grouped by the minimum number of negative $\omega$ values among the shortest paths connecting each node to the origin node.  We see that the fewer such negative $\omega$ values the faster the arrival times and most (in this example $91\%$) of the nodes have a path connecting them to the origin node with two or less negative $\omega$ values.  

{\color{black} We also considered the effect of different infection rates in the southern versus northern hemispheres.  In the airline network taken from \cite{openflights}, only about $20\%$ of the airports reside in the southern hemisphere.  Some numerical results are presented in Figure~\ref{fig:southnorth}.  First we consider the case where the infection rate is greater in the southern than northern hemisphere.  This causes arrival times in most of the network to be advanced relative to the values predicted when the infection rate is constant and equal to the global mean.  If the prediction is changed to instead use the infection rate for the southern hemisphere then the predicted versus observed arrival times is almost linear for cities in the southern hemisphere owing to the fact that most pairs of cities in the southern hemisphere are connected by shortest paths visiting only other cities in the southern hemisphere.   When the infection rate is greater in the northern hemisphere a similar dynamic occurs and arrival times in the northern hemisphere are advanced and approximately linear.  Since $\omega_n$ is rather small in the northern hemisphere this advancement is not as dramatic as it is for larger infection rates in the southern hemisphere.     }

\begin{figure}
    \centering
     \subfigure{\includegraphics[width=0.33\textwidth]{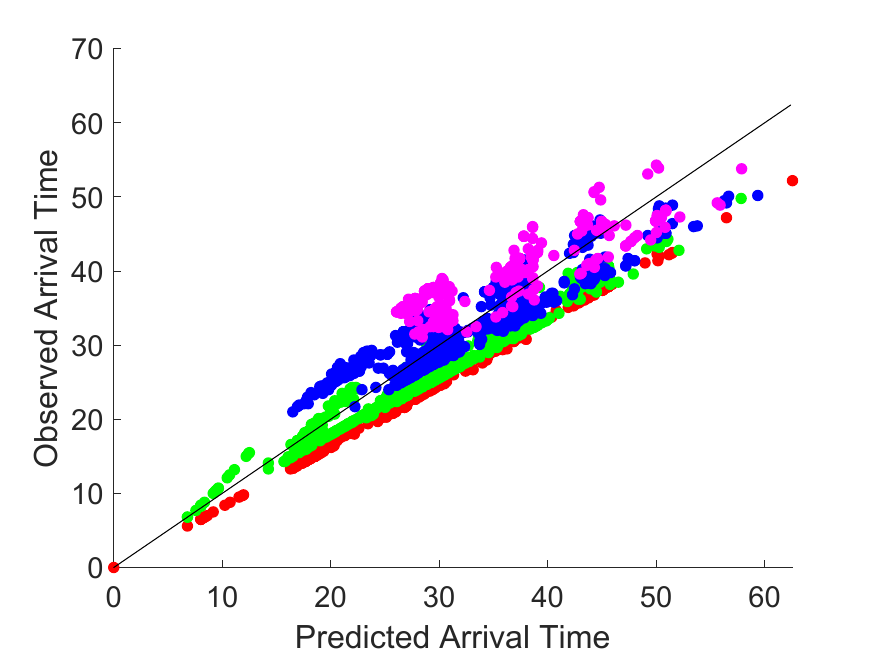}}
 \subfigure{\includegraphics[width=0.33\textwidth]{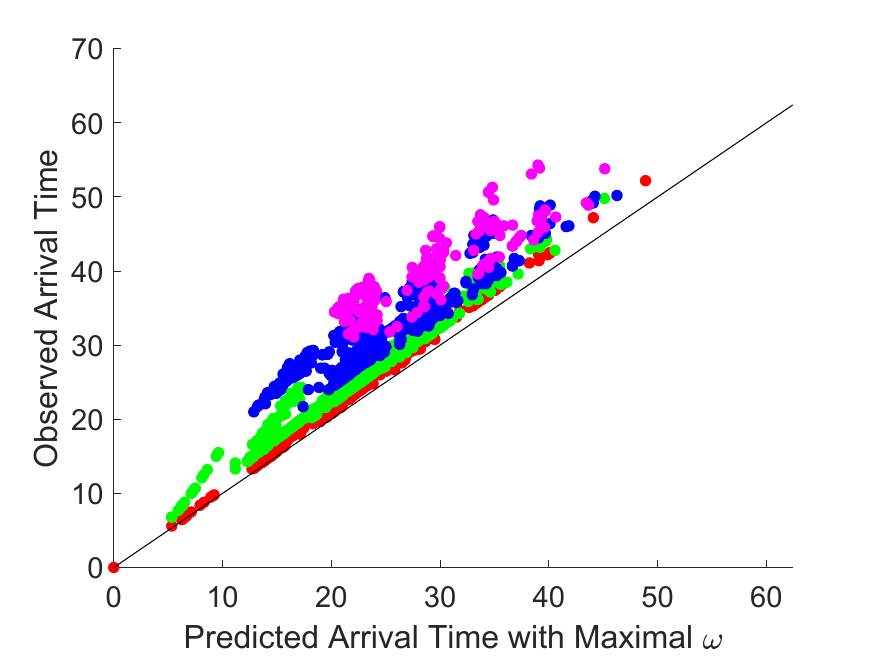}}
   \caption{Arrival times for (\ref{eq:inhomo}) on the worldwide airline transportation network with half the nodes assigned $\omega_n=0.2$ uniformly at random with the remaining nodes having $\omega=-0.2$. On the left, arrival times are plotted against the linear prediction for the mean value of $\alpha=1.0$ ($\beta=0.25$ and $\gamma=0.001$).  Consider all paths that connect a node $m$ to the origin node with the minimal graph distance $d_m$.  The data points in red are those for which there exists a minimal path on which all $\omega_n>0$.  
  Green corresponds to nodes with a minimal path with exactly one negative $\omega_n$.  Blue nodes have two negative $\omega$ values while magenta has three.  The arrival times of all red nodes are advanced in the inhomogeneous system.  Around $95\%$ of the nodes for which there exists a minimal path with exactly one $\omega_n<0$ arrive faster (green nodes)  and around $67\%$ of the nodes with minimal paths with exactly two $\omega_n<0$ arrive faster (blue nodes).   This covers $91\%$ of the total nodes in the network.    On the right, we compare arrival times in (\ref{eq:main}) with the linear prediction (\ref{eq:linATW}) assuming that all nodes have $\alpha=1.2$.  Observe that this constitutes a reasonable prediction for the arrival times at nodes with all a path of all positive $\omega$ values (red data points).    }
    \label{fig:inhomoexplanation}
\end{figure}

\begin{figure}
    \centering
     \subfigure{\includegraphics[width=0.3\textwidth]{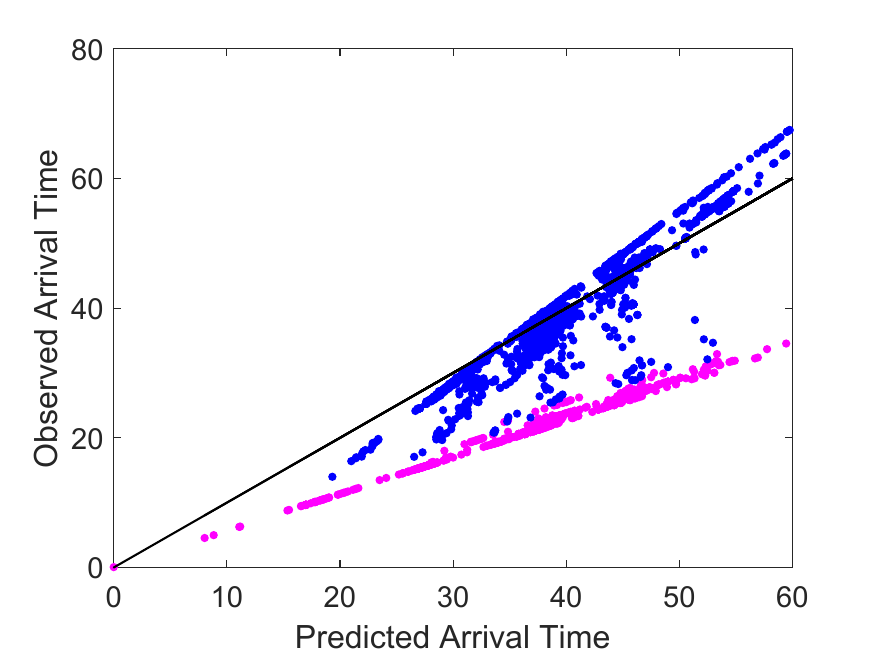}}
 \subfigure{\includegraphics[width=0.3\textwidth]{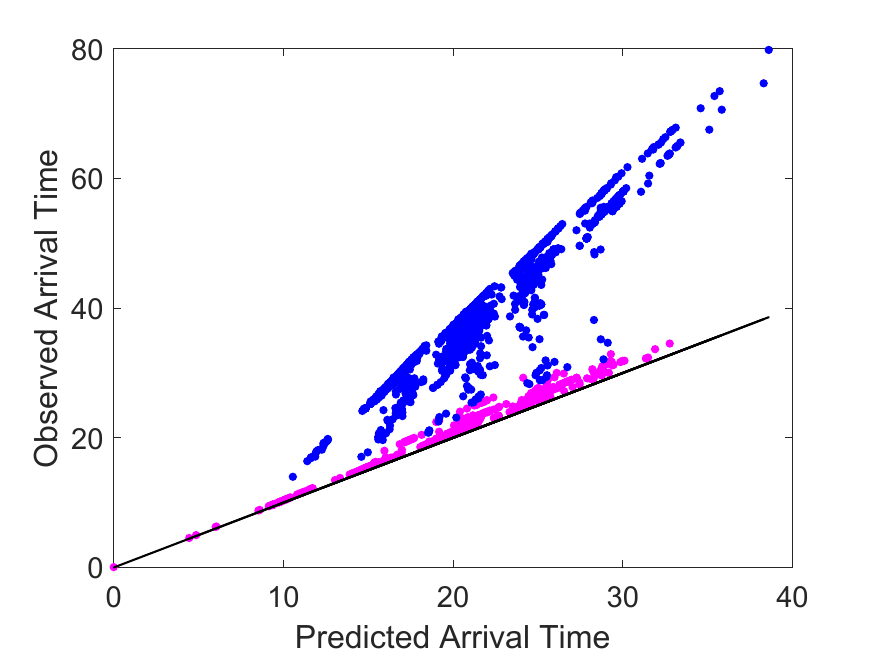}}
  \subfigure{\includegraphics[width=0.3\textwidth]{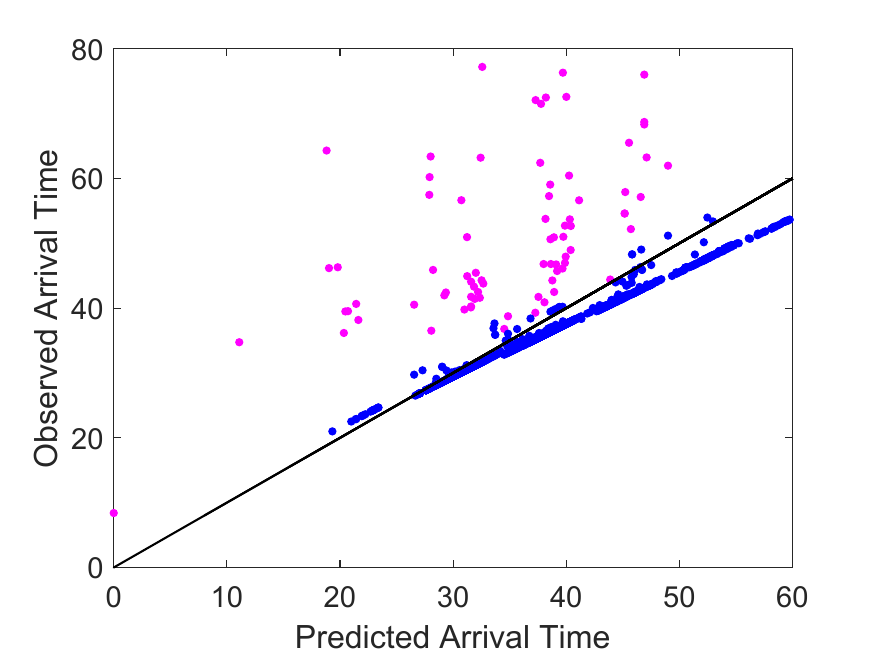}}
   \caption{Arrival times versus predictions for the worldwide airline transportation network \cite{openflights} with $\omega_n$ selected by hemisphere.  In all simulations, the mean infection rate is fixed to $\alpha=1.0$, the recovery rate is fixed to $\beta=0.25$ and the diffusion parameter is fixed to $\gamma=0.001$.  In all simulations the original city of infection resides the southern hemisphere.  On the left, $\omega_n>0$ for those airports in the southern hemisphere and $\omega_n<0$ for those airports in the northern hemisphere.  The purple data points are arrival times for cities in the southern hemisphere while the blue dots are arrival times for airports in the northern hemisphere.  In the left panel, the predicted arrival time is the linear arrival time estimate (\ref{eq:linAT}) with $\alpha$ fixed to be the mean infection rate.  In the middle panel, the predicted arrival time is instead the the linear arrival time estimate with the maximal infection rate (constant in the southern hemisphere).  On the right, we show arrival times for the case where $\omega_n>0$ in the northern hemisphere while $\omega_n<0$ in the southern hemisphere.  The predicted arrival time is the linear arrival time estimate (\ref{eq:linAT}) with $\alpha$ fixed to be the mean infection rate. }
    \label{fig:southnorth}
\end{figure}

\begin{rmk} It is known in the PDE context that inhomogeneities can lead to faster invasion speeds; see for example \cite{berestycki19,skt86}. In these cases the system typically exhibits pulsating traveling waves that propagate with some mean velocity that exceeds the velocity in the homogeneous case.  We emphasize that the mechanism at play in the PDE case is distinct than the one we discuss here.  
\end{rmk}

\section{Conclusion}\label{sec:discussion}

We have illustrated that the analogy between the dynamics of the meta-population model (\ref{eq:main}) and invasion fronts for spatially extended reaction-diffusion systems can be used to make qualitative predictions on the behavior of (\ref{eq:main}) in certain circumstances.  To recap, we show that arrival time estimates can be procured for a variation of (\ref{eq:main}) that includes an exposed population.  Second, from the PDE theory we expect that faster than linear invasion speeds should arise for some models where the nonlinearity enhances the growth of the instability.  {\color{black} Using a model motivated by recent work on the role of higher-order interactions in social epidemics we demonstrate that this also occurs in the case of the meta-population model (\ref{eq:main}).  } Using the smallness of the diffusion constant $\gamma$ and viewing the invasion front as a cascading process we are able to obtain corrections to the linear arrival times that provide more accurate predictions of arrival times.  Finally, we consider the effect of inhomogeneities on mean arrival times.  In the PDE case this can lead to faster arrival times.  We show that the same phenomena occurs in (\ref{eq:main}) although we argue that the mechanism leading to the decrease is distinct and due to the asymmetry between how local growth rates and random walk probabilities affect the arrival time calculation.  

{\color{black}
We conclude with comments on some directions for further research.  

Throughout this article we have assumed that the mobility parameter $\gamma$ is asymptotically small.  This assumption is valid in some situations, but it would be valuable to understand how arrival times are determined for larger values of $\gamma$.  This could be relevant when mobility is increased or when the infection rate is only slightly larger than the recovery rate so that the homogeneous growth and diffusion terms have similar scalings.  Numerical simulations of the SIR model (\ref{eq:main}) suggest that linear arrival times remain good estimates for nonlinear arrival times even for larger values of $\gamma$; see Figure~{\ref{fig:biggamma}}.  It would be interesting if it were possible to characterize which network features are relevant for this decreased arrival times.  

In terms of mathematical analysis, it would be interesting to establish rigorous upper bounds on nonlinear arrival times to complement the lower bounds afforded by the linearized equation in Theorem~\ref{thm:super}.  One possible avenue is to derive sub-solutions for (\ref{eq:main}).  We refer to \cite{fu16,wu17} for work in this direction for lattice SIR models.   We have used the term linearly determined informally to describe situations where the linearized arrival times are good estimates for the nonlinear arrival times.  A rigorous bound on nonlinear arrival times would serve to make this mathematically precise.  We point to recent work characterizing the location of solution level sets for the lattice Fisher-KPP equation as a starting point for this analysis; see \cite{besse22}.  

Several qualitative predictions for how network and system properties determine arrival times in meta-population models of global disease spread have been presented.  Ultimately, part of the motivation of the current study was to provide predictions that might be applied to more complicated and realistic models of disease spread.

\begin{figure}
    \centering
     \subfigure{\includegraphics[width=0.33\textwidth]{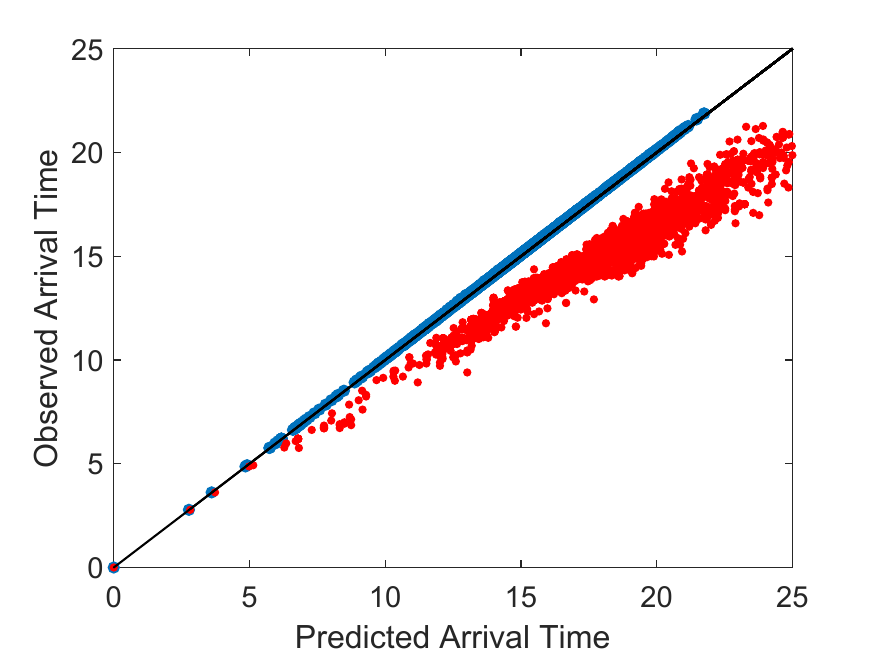}}
 \subfigure{\includegraphics[width=0.33\textwidth]{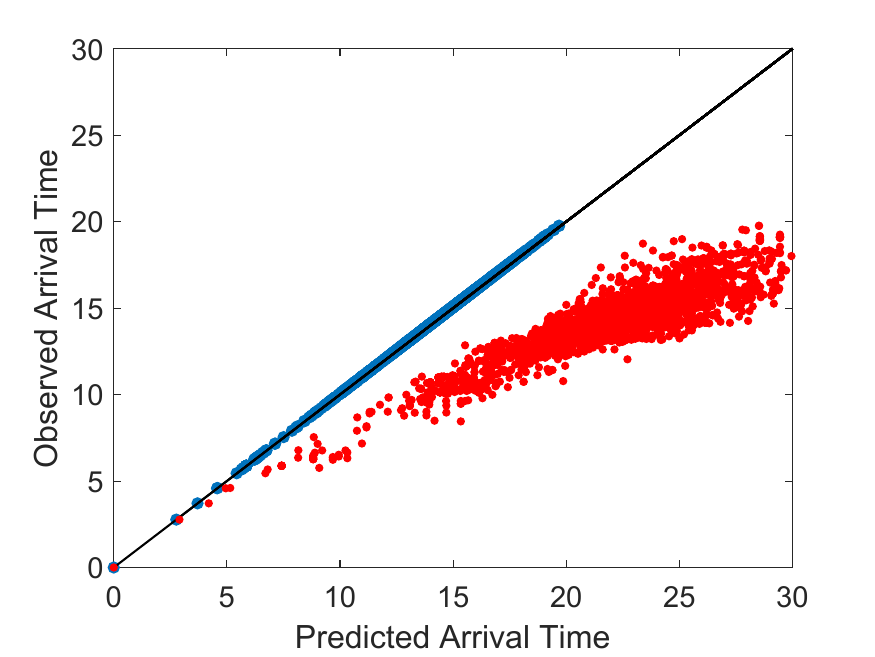}}
   \caption{Predicted versus observed arrival times for the SIR model (\ref{eq:main}) with infection rate $\alpha=1.0$, recovery rate $\beta=0.25$ and diffusion parameter $\gamma=0.3$ (left) and $\gamma=0.5$ (right).  The red predictions are those given by (\ref{eq:linATW}) which uses only the first term in the summation (\ref{eq:linAT}) while  the blue data are predictions computed by numerically solving (\ref{eq:linAT}) including the first twenty non-zero terms in the summation.  As expected, the one term approximation over-estimates the arrival times since it considers only contributions coming from the shortest path.   The correspondence between the linear prediction and nonlinear arrival times observed in numerical simulations suggests that (\ref{eq:main}) remains linearly determined even for large values of $\sigma$, although no closed form expression for arrival times is available.  }
    \label{fig:biggamma}
\end{figure} 

}
 \begin{appendix}
{\color{black} 
 \section{Singular perturbation analysis of the local model \ref{eq:SIRsimplicial}} \label{sec:appendix} 
 We consider (\ref{eq:SIRsimplicial}) with the goal of motivating the approximate solution presented in (\ref{eq:Ipushed}).  Our approach mimics the analysis of a model of an autocatalator chemical reaction model presented in \cite{gucwa09}.   We begin with the system (\ref{eq:T1}) where we wish to track the solution to the initial value problem with initial conditions $S(0)=1-\kappa$, $I(0)=\kappa$ in the limit as $\e=\frac{1}{\rho}\to 0$.  As mentioned in Section~\ref{sec:pushed} this system has two slow manifolds defined as curves of equilibrium when $\e$ is set equal to zero; see (\ref{eq:T1fast}).  The slow manifold on the $I$ axis is normally hyperbolic and it follows that the reduced flow on the slow manifold is, to leading order in $\e$ given by $I'=-\beta I$ and so we obtain that after some critical time $\Omega$ the solution of $I(t)$ can be described as in (\ref{eq:Ipushed}).  The second slow manifold is given by the $S$ axis, but this manifold lacks normal hyperbolicity so we are unable to track the solution of the initial value problem using linearization.  
 
 To overcome this lack of normal hyperbolicity we use geometric desingularization techniques or ``blow-up'' techniques to resolve the flow when $I$ is small.  Following \cite{gucwa09} we will change  coordinates to
 \[ S=\bar{S}, \ I=r \bar{I}, \ \e=r\bar{\e}, \ \bar{I}^2+\bar{\e}^2=1, \]
effectively transforming the $S$ axis to a cylinder with polar coordinates for the $I$ and $\e$ variables.  It is often easier to study the flow in coordinate charts and we employ two distinct charts.  The first is known as the re-scaling chart with coordinates
\[ S=S_1, \ I=r_1 I_1, \ \e=r_1, \]
while the second chart has coordinates 
\[ S=S_2, \ I=r_2, \ \e=r_2\e_2.\]
The two charts can be related via
\[ S_2=S_1, \ r_2=r_1I_1, \ \e_2=\frac{1}{I_1}. \]
Our goal is to track an initial condition with $S(0)=1-\kappa$, $I(0)=\kappa$ with $\kappa$ small as it evolves past the non-hyperbolic $S$ axis to the section $\Sigma_{out}=\{ (S,I) \ | \ I=\eta\}$ for some $\eta>0$ at which the solution can be effectively described by a fast transition to the $I$ axis followed by a slow relaxation along the $I$ axis until the solution converges to the origin.  In contrast to \cite{gucwa09}, our estimates here are approximate and not rigorous.  We believe that the estimates presented here could be made rigrorous, but we do not pursue such an analysis here.  
\paragraph{Analysis in first chart} 
The first chart is known as the rescaling chart where $r_1$ is simply a proxy for $\e$.  Converting (\ref{eq:T1}) to the coordinates of the first chart we find,
\begin{eqnarray}
\frac{dS_1}{d\tau}&=& -\alpha r_1^2 S_1 I_1-S_1r_1^2I_1^2 \nonumber \\
\frac{dI_1}{d\tau} &=& \alpha r_1S_1I_1 -\beta r_1 I_1+S_1r_1 I_1^2 \nonumber \\
\frac{dr_1}{d\tau}&=& 0 \label{eq:chartone} 
\end{eqnarray}
Rescaling the independent variable to divide the vector field by $r_1$ we find the de-singularized system 
\begin{eqnarray}
\frac{dS_1}{dt}&=& -\alpha r_1 S_1 I_1-S_1r_1I_1^2 \nonumber \\
\frac{dI_1}{dt} &=& \alpha S_1I_1 -\beta  I_1+S_1 I_1^2 \nonumber \\
\frac{dr_1}{dt}&=& 0. \label{eq:chartonedesing} 
\end{eqnarray}
Let $\eta>0$ and define the section $\Sigma_1=\{ (S_1,I_1,r_1) \ | \ I_1=\eta\} $.  Suppose that we start with initial conditions $I(0)=\kappa$ and $S(0)=1-\kappa$ which correspond to initial conditions $S_1(0)=1-\kappa$ and $I_1(0)=\frac{\kappa}{\e}$.  We therefore require $\kappa$ to scale smaller than $\e$ so that $I_1(0)$ is near zero.  To obtain a leading order description of the dynamics we set $r_1=0$ in (\ref{eq:chartonedesing}) and approximate $S_1(t)=1$.  Then $I_1$ obeys (to leading order in $\e$)
\[ \frac{dI_1}{dt} = (\alpha-\beta)I_1 + I_1^2, \quad I_1(t)=\frac{C(\alpha-\beta)e^{(\alpha-\beta)t}}{1-Ce^{(\alpha-\beta)t}}, \ C=\frac{\kappa}{\kappa+\e (\alpha-\beta)}. \]
Define $\Omega_1$ such that $I_1(\Omega_1)=\eta$.  Using the leading order description for $I_1(t)$ we estimate 
\[ \Omega_1\approx\frac{1}{\alpha-\beta} \log \left( \frac{\eta (\kappa+\e(\alpha-\beta))}{(\alpha-\beta+\eta)\kappa}\right) \]


We now convert our solution to the coordinates of the second chart and proceed with tracking the solution.  
\paragraph{Analysis in second chart} 
Converting (\ref{eq:T1}) to the coordinates of the second chart we find,
\begin{eqnarray}
\frac{dS_2}{d\tau}&=& -\alpha r_2^2\e_2 S_2-r_2^2S_2 \nonumber \\
\frac{dr_2}{d\tau} &=& \alpha r_2^2\e_2 S_2 -\beta r_2^2 \e_2+S_2r_2^2 \nonumber \\
\frac{d\e_2}{d\tau}&=& -\alpha r_2\e_2^2 S_2 +\beta r_2 \e_2^2-S_2r_2\e_2 \label{eq:charttwo} 
\end{eqnarray}
Rescaling the dependent variable to divide the vector field by the non-zero factor $\alpha r_2\e_2S_2-\beta r_2\e_2+S_2r_2$ we obtain the desingularized system
\begin{eqnarray}
\frac{dS_2}{ds}&=& -r_2\left(\frac{1}{1-\frac{\beta\e_2}{\alpha\e_2S_2 +S_2}}\right) \nonumber \\
\frac{dr_2}{ds} &=& r_2 \nonumber \\
\frac{d\e_2}{ds}&=& -\e_2. \label{eq:charttwo} 
\end{eqnarray}
Define $\Sigma_2=\{ (S_2,r_2,\e_2) \ | \ r_2=\eta\}$   with $\eta$ defined as before and recall the initial conditions in the section $\Sigma_1$ which correspond to $S_2(0)=1-\kappa+\O(\e)$, $r_2(0)=\eta\e $, $\e_2(0)=\frac{1}{\eta}$.  
The transition time between sections can then be evaluated explicitly, it terms of the transformed time-scale $s$, as $s=-\log \e$.  To determine estimates for the transition time $\tau_2$ in the $\tau$ time-scale we note that the timescales are related by the integral
\[ \tau_2=\int_0^{-\log(\e)} \frac{1}{\alpha r_2\e_2 S_2(\sigma)-\beta r_2\e_2 +S_2(\sigma)r_2(\sigma)} d\sigma \]
We will obtain an approximation to $t_2$ by setting $S_2(\sigma)=1$ in the integral.  We are then able to integrate (recalling that $r_2\e_2=\e$) and find
\begin{eqnarray*} \tau_2 &\approx& \frac{1}{\e(\alpha-\beta)}\left( -\log(\e)+\log\left(\frac{\e(\alpha-\beta)+\e\eta}{\e(\alpha-\beta)+\eta}\right)\right), \\
&\approx& \frac{1}{\e(\alpha-\beta)}\left( -\log(\e)+\log\left(\frac{\e}{\eta}\frac{(\alpha-\beta)+\eta}{1+\frac{\e}{\eta}(\alpha-\beta)}\right)\right)\\
&\approx& \frac{1}{\e(\alpha-\beta)}\left( -\log(\eta)+\log\left(\frac{\alpha-\beta+\eta}{1+\frac{\e}{\eta}(\alpha-\beta)}\right)\right)
\end{eqnarray*}
Re-scaling the independent variable from $\tau$ to $t$ we obtain an estimate on the total transit time of the initial condition $I(0)=\kappa\e$ to $I(t)=\eta$ as
\begin{eqnarray*} \Omega &\approx& \frac{1}{\alpha-\beta}\left( \log \left(\eta \frac{ (\kappa+\e(\alpha-\beta))}{(\alpha-\beta+\eta)\kappa}\right) -\log(\eta)+ \log\left(\frac{\alpha-\beta+\eta}{1+\frac{\e}{\eta}(\alpha-\beta)}\right)\right) \\
&\approx &  \frac{1}{\alpha-\beta}\left( \log \left(\frac{ (\kappa+\e(\alpha-\beta))}{\kappa}\right) +\log\left(\frac{1}{1+\frac{\e}{\eta}(\alpha-\beta)}\right)\right).
\end{eqnarray*}
Using $\frac{\kappa}{\e}$ small and $\e\ll 1$ we find the approximation in (\ref{eq:Omega}). }

 \end{appendix}

\section*{Acknowledgements} 
This project was conducted as part of a NSF sponsored REU program.  All participants received support from the NSF (DMS-2007759).  The authors thank the anonymous referees whose comments improved the paper.

\section*{Data Availability}  The datasets generated during and/or analysed during the current study are available from the corresponding author on reasonable request.

\bibliographystyle{abbrv}
\bibliography{REUMaster}

\end{document}